\documentclass[11pt]{article}
\usepackage{amssymb,amsmath,amsthm,amsfonts}
\usepackage{graphicx}
\usepackage{subfigure}
\usepackage[usenames, dvipsnames]{color}
\usepackage{epstopdf}
\usepackage{mathrsfs}
\usepackage[CJKbookmarks, colorlinks, bookmarksnumbered=true,pdfstartview=FitH,linkcolor=black]{hyperref}
\usepackage{tcolorbox}
\def\disp{\displaystyle}

\def\crr{\cr\noalign{\vskip2mm}}

\def\dref#1{(\ref{#1})}

\theoremstyle{plain}
\newtheorem{theorem}{Theorem}[section]
\newtheorem{lemma}{Lemma}[section]
\newtheorem{proposition}{Proposition}[section]
\newtheorem{corollary}{Corollary}[section]

\numberwithin{equation}{section}

\theoremstyle{definition}
\newtheorem{definition}{Definition}

\newtheorem{remark}{Remark}[section]

\newcommand{\R}{{\mathbb R}}

\def\A{\mathcal{A}}

\begin{document}
\title{{\bf
Dynamics Compensation in  Observation of Abstract   Linear Systems\footnote{\small
This work is supported by  the National Natural Science
Foundation of China  (Nos. 61873153,  61873260). }} }

\author{ Hongyinping Feng$^{a}$\footnote{\small  Corresponding author.
Email: fhyp@sxu.edu.cn.}, \  Xiao-Hui Wu$^{a}$\ and \ Bao-Zhu Guo$^{b,c}$\\
$^a${\it \small School of Mathematical Sciences,}
{\it \small Shanxi University,  Taiyuan, Shanxi, 030006,
China}\\
$^b${\it \small Department of Mathematics and Physics,  }
{\it \small  North
China Electric Power University, Beijing 102206, China}\\
$^c${\it \small Key Laboratory of System and Control, Academy of Mathematics and Systems Science,}
\\{\it \small  Academia Sinica, Beijing, China}
 }

 \maketitle

\begin{abstract}
This is the second part of four  series papers, aiming at the problem of
sensor dynamics compensation for abstract linear systems. Two major issues are addressed. The first
 one is about  the sensor dynamics compensation in    system observation
 and the second one
  is on the disturbance dynamics compensation in   output regulation  for linear system.
Both of them can be described   by the problem of state  observation  for  an abstract  cascade system.
We consider these two  apparently different problems  from the same abstract linear system point of view.
A new scheme of the observer design for the  abstract  cascade system
  is developed and the exponential convergence of the  observation error is established.
It is shown that the error based  observer design in the problem of   output regulation can be converted into a  sensor dynamics compensation problem
 by the well known   regulator  equations.
As a  result, a tracking error based   observer for output regulation problem  is designed by exploiting the
  developed method.
 As   applications,
  the ordinary differential  equations (ODEs)
with output time-delay and an unstable heat equation with ODE  sensor dynamics are fully investigated
to validate the   theoretical results.
 The numerical simulations for the unstable heat  system are carried out to validate the
 proposed method visually.

\vspace{0.3cm}


\vspace{0.3cm}

\end{abstract}
\section{Introduction}

When a sensor is installed on  a control plant  indirectly,
the dynamics  that connect   the control plant  and the sensor are referred to as  sensor dynamics.
In this case, one has to compensate the  sensor dynamics in the observer  design.
The compensation of   sensor dynamics  dominated by partial differential equations (PDEs) has   attracted much attention in recent years.
  The most commonly used   example for infinite-dimensional sensor dynamics compensation  is the  output time-delay compensation in the context of the backstepping method (\cite{SmyshlyaevKrstic2008SCL}),    where the time-delay is regarded as the  sensor dynamics dominated by a transport equation.
  Soon after   \cite{SmyshlyaevKrstic2008SCL}, the  sensor delay dynamics have been extended to the dynamics dominated by the first order hyperbolic equation     \cite{Ferrantea2020Aut},  heat equation \cite{Krstic2009SCL,TangSCL2011} and   the wave equation \cite{Krstic2009TAC}.
All are by the  PDE backstepping method (\cite{KrsticDelaybook}).

The PDE backstepping method is    powerful  in the observer design for  ODE-PDE cascade systems.
However, it relies on a priori target system before the backstepping transformation can be performed. This means that the successful employment of the backstepping method requires a proper choice of target system.
Although  such choice seems natural  in many situations like those in   \cite{SmyshlyaevKrstic2008SCL,Krstic2009SCL,Krstic2009TAC},  it is more or less   relying  on intuition rather than strict  analysis.
This limits the applicability of the PDE backstepping method and can be seen
from the systems described by   Euler-Bernoulli beam equations or multi-dimensional
PDEs for which the backstepping method is hardly  applied.
Roughly speaking, a  sufficient condition  that  ensures  the applicability of the PDE backstepping method is still lacking.
Moreover,  in the PDE backstepping method, the Lyapunov function has to  be constructed
in the stability analysis for the resulting closed-loop system, which gives rise to some challenges  in particular  for the treatment of  PDEs with delays.

System observation through sensor dynamics can usually be modeled as a    cascade  observation  system.
Similarly, in   output  regulation for  linear  systems,  the control plant driven by the  disturbance generated from an exosystem  can also be described by a cascade system.
This implies  that the sensor dynamics compensation is closely related to the   observer design for estimation of the disturbance and the state simultaneously in the output regulation, although these two problems were investigated separately in literature.
As will be seen, these two problems can be connected equivalently by an invertible transformation constructed by the well known  regulator  equations in the output regulation.  As a consequence, an  error based observer can be designed for  the output regulation of    abstract  systems.

In this paper, we unify  various  types of sensor dynamics compensations from a general abstract framework point of view.
We model  the control system  with indirect sensor configuration as a cascade system.
Let $X_j$, $U_j$, and $Y_j$   be  Hilbert spaces and let
$A_j:X_j\to X_j$, $B_j:U_j\to X_j$,   $C_1:X_1\to Y_1$ and $C_2:X_2\to U_1$
be  related operators which are possibly unbounded, $j=1, 2$.
The considered  cascade system is described by
\begin{equation} \label{wxh20202131701}
 \left\{\begin{array}{l}
\disp \dot{x}_1(t) = A_1 x_1(t)+B_1C_2x_2(t) ,\crr
\disp \dot{x}_2(t) = A_2 x_2(t)+B_2u(t)  ,\crr
\disp y(t)=C_1x_1(t),
\end{array}\right.
\end{equation}
where  $y(t)$ is the  measured output,  $u(t)$ is the control and $B_1C_2:X_2\to X_1$   represents the interconnection.
The $x_2$-subsystem of \dref{wxh20202131701} is the control plant but the sensor is installed indirectly  on the  $x_1$-subsystem.

The aim  of this paper is to develop a systematic methodology for  the observer design  of  system \dref{wxh20202131701}.
The observer  we shall design  is   the well known Luenberger-like observer.
Owing to the  generality of   system \dref{wxh20202131701}, the  sensor dynamics  dominated by either  ODE or PDE  like the first order hyperbolic equation,  heat equation, wave equation
and the Euler-Bernoulli beam equation can be compensated effectively  in a  unified framework.
When  $A_2$ is a matrix, the existence of the observer gains  can be characterized by  the system \dref{wxh20202131701}  itself.
More importantly, the more complicated problem that   observation of  the infinite-dimensional system through finite-dimensional sensor dynamics can still be addressed effectively.
To demonstrate the effectiveness of the proposed approach,  the observer design  for an  unstable heat equation with   ODE sensor dynamics  is considered.

We organize the paper as follows. In Section \ref{motivation}, we
demonstrate  the  main idea of   sensor dynamics compensation  through
a finite-dimensional cascade system.
Some  preliminary results on abstract systems are presented  in Section  \ref{preliminary}.
The well-posedness of the open-loop system is discussed in Section \ref{Se.4}.
The observability is discussed in Section \ref{Observability}.
Section \ref{Observer1} is devoted to  the Luenberger-like observer design  for   abstract linear systems, where the results in Section \ref{motivation} for the finite-dimensional system
are extended to the infinite-dimensional counterparts.
The well-posedness and the exponential convergence  of the observer   are obtained.
As an application,  the disturbance and state  observation in output regulation is investigated
in Section \ref{secDisturbance}.
To validate the theoretical results,  we consider the observations of ODEs
with output  delay in  Section  \ref{ODE+delay}  and an unstable heat system with ODE sensor  dynamics in     Section \ref{heat+ODE}.
Some numerical simulations are presented in Section \ref{simulation} to show the theory visually,
followed up conclusions in Section \ref{remarksensor}.
For the sake of  readability, some results which  are less relevant to the dynamics compensator design are arranged  in the Appendix.

Throughout the paper, the inner product of the Hilbert space $X$  is denoted by $\langle\cdot,\cdot\rangle_{X}$, and the induced  norm is denoted by $\|\cdot\|_{X}$.
Identity operator on the Hilbert space $X_i$ is denoted by $I_i, i=1,2$. The space of bounded linear operators from  $X_1$ to $X_2$ is denoted by $\mathcal{L}(X_1, X_2)$.
If $A\in \mathcal{L}(X_1, X_2)  $, we represent  the kernel,   domain,
    resolvent set and the spectrum of $A$ as  ${\rm Ker} (A)$,  $D (A)$,
$\rho(A)$ and  $\sigma(A)$, respectively.

\section{Motivation }\label{motivation}
As mentioned, this section demonstrates  the  main idea of the observer design for   system  \dref{wxh20202131701} in the  setting of  finite-dimensional state  space.
The Luenberger observer  of system  \dref{wxh20202131701} is designed as
\begin{equation} \label{wxh20202061740}
 \left\{\begin{array}{l}
\disp \dot{\hat{x}}_1(t) =  A_1\hat{x}_1(t)+B_1C_2 \hat{x}_2 (t)-F_1[y(t)- C_1\hat{x}_1(t)],\crr
\disp  \dot{\hat{x}}_2(t) =A_2\hat{x}_2(t)+F_2[y(t)-C_1\hat{x}_1(t)]+B_2u(t),
\end{array}\right.
\end{equation}
where $F_j\in \mathcal{L}(Y_1,X_j), j=1,2$ are the gain vectors  to be determined.
When system  \dref{wxh20202131701} is observable, $F_1$ and $F_2$ can be chosen easily by
the pole assignment theorem of the  linear  systems.
However,  the problem becomes complicated if  either $A_1$ or  $A_2$ is an operator in   an infinite-dimensional Hilbert space.
We therefore need  an alternative way to find $F_1$ and $F_2$ to be adaptive to the infinite-dimensional setting.
Let the  observer errors be
\begin{equation} \label{wxh20202061802}
\tilde{x}_j(t)=x_j(t)-\hat{x}_j(t),\ \ j=1,2,
\end{equation}
which  are  governed by
\begin{equation} \label{wxh20202061803}
 \left\{\begin{array}{l}
\disp \dot{\tilde{x}}_1(t) = (A_1+F_1C_1)\tilde{x}_1(t)+B_1C_2 \tilde{x}_1(t),\crr
\disp  \dot{\tilde{x}}_2(t) =  A_2 \tilde{x}_2(t)-F_2C_1\tilde{x}_1(t).
\end{array}\right.
\end{equation}
If we  select  $F_1$ and $F_2$ properly such that  system \dref{wxh20202061803} is stable, then
$(x_1(t),x_2(t))$  can be estimated in the sense that
\begin{equation} \label{wxh20202062106}
  \|(x_1(t)-\hat{x}_1(t) ,{x}_2(t)-\hat{x}_2(t))\|_{X_1\times X_2}\to  0
  \ \ \mbox{as}\ \ t\to\infty.
\end{equation}
Inspired by the first part \cite{FPart1} of this series works, the $F_1$ and $F_2$ can be chosen
easily by decoupling  the   system \dref{wxh20202061803} as a cascade system.
The corresponding transformation is
\begin{equation} \label{wxh20202061804}
\begin{array}{l}
\begin{pmatrix}
I_1&S\\
0&I_2
\end{pmatrix}
\begin{pmatrix}
A_1+F_1C_1&B_1C_2 \\
-F_2C_1&A_2
\end{pmatrix}\begin{pmatrix}
I_1&S\\
0&I_2
\end{pmatrix}^{-1}\crr
=
\begin{pmatrix}
A_1+(F_1-SF_2)C_1& SA_2-[A_1+(F_1-SF_2)C_1]S+B_1C_2 \\
-F_2C_1&A_2+F_2C_1S
\end{pmatrix},
\end{array}
\end{equation}
where $S\in \mathcal{L}( X_2,X_1)$ is to be determined.
If we select $S$ properly such that
\begin{equation} \label{20208292004}
SA_2-[A_1+(F_1-SF_2)C_1]S+B_1C_2=0,
\end{equation}
  then   the right side matrix of \dref{wxh20202061804} is Hurwitz if and only if the matrices $A_1+(F_1-SF_2)C_1$ and $A_2+F_2C_1S$ are Hurwitz.

\begin{lemma}\label{wxh20202061806}
Let    $X_1, X_2$, $U_1$, $U_2$ and $ Y_1$ be Euclidean spaces, and let $A_j\in \mathcal{L}(X_j)$,    $B_j\in  \mathcal{L}(U_j ,X_j)$, $j=1,2$,  $ C_2  \in \mathcal{L}(X_2,U_1)$ and $C_1\in \mathcal{L}(X_1,Y_1)$.
Suppose that system \dref{wxh20202131701} is observable.
Then, there exist $F_1\in \mathcal{L}(Y_1,X_1)$ and $F_2\in \mathcal{L}(Y_1,X_2)$ such that the solution of the observer \dref{wxh20202061740} satisfies:
\begin{equation} \label{wxh20202061808}
  \|(x_1(t)-\hat{x}_1(t), x_2(t)-\hat{x}_2(t))\|_{X_1\times X_2}\to  0
  \ \ \mbox{as}\ \ t\to\infty.
\end{equation}
Moreover, $F_1$ and $F_2$ can be selected  by the  scheme:
(a)~ Select  $F_0\in \mathcal{L}(Y_1,X_1)$ such that $A_1+F_0C_1$ is  Hurwtiz;
(b)~  Solve the Sylvester equation $(A_1+F_0C_1)S-SA_2=B_1C_2 $;
(c)~ Select  $F_2\in \mathcal{L}(Y_1,X_2)$ such  that  $A_2+F_2C_1S$ is  Hurwtiz;
(d)~  Set $F_1=F_0+SF_2$.
\end{lemma}

\begin{proof}
Since  system \dref{wxh20202131701} is observable, its cascade structure  implies that
$(A_1,C_1)$ is observable as well.
As a result, there exists an   $F_0\in \mathcal{L}(Y_1,X_1)$ such that $A_1+F_0C_1$ is Hurwtiz,
and at the same time,
 \begin{equation} \label{wxh20202061807}
\sigma(A_1+F_0C_1)\cap\sigma(A_2)=\emptyset.
\end{equation}
 By \cite{Rosenblum1956DMJ}, the Sylvester equation
  \begin{equation} \label{wxh20202061805}
(A_1+F_0C_1) S -SA_2=B_1C_2
\end{equation}
admits a unique solution $S\in \mathcal{L}( X_2,X_1)$.
Hence,  the invertible transformation  in \dref{wxh20202061804} satisfies
\begin{equation} \label{20208111108}
\begin{array}{l}
\begin{pmatrix}
I_1&S\\
0&I_2
\end{pmatrix}
\begin{pmatrix}
A_1+F_0C_1&B_1C_2 \\
0&A_2
\end{pmatrix}\begin{pmatrix}
I_1&S\\
0&I_2
\end{pmatrix}^{-1}
=
\begin{pmatrix}
A_1+F_0C_1& 0 \\
0&A_2
\end{pmatrix}
\end{array}
\end{equation}
and
\begin{equation} \label{20208111112}
( C_1 , 0) \begin{pmatrix}
I_1&  S\\
0 &I_2
\end{pmatrix}^{-1}=
 ( C_1 , - C_1S) .
\end{equation}
Since the observability of system \dref{wxh20202131701} keeps invariant  under the output feedback,
system \dref{wxh20202131701} is observable if and only if the pair
\begin{equation} \label{wxh20202071746}
\left(\begin{pmatrix}
 A_1+F_0C_1 & 0\\
0 &A_2
\end{pmatrix},
 ( C_1 , - C_1S)\right)
\end{equation}
is observable.
By  \dref{wxh20202061807},  the observability of \dref{wxh20202071746} implies that $(A_2, C_1S)$ is observable \cite[Lemma 10.2]{FPart1}.
Therefore, there exists an  $F_2\in \mathcal{L}(Y_1,X_2)$ such that
$ A_2+F_2C_1S$ is Hurwitz.
Thanks to the  scheme of choosing  $F_1$, equation \dref{wxh20202061805} becomes  \dref{20208292004}. As a result,  \dref{wxh20202061804} becomes
\begin{equation} \label{wxh20202062100}
\begin{array}{l}
\begin{pmatrix}
I_1&S\\
0&I_2
\end{pmatrix}
\begin{pmatrix}
A_1+F_1C_1&B_1C_2 \\
-F_2C_1&A_2
\end{pmatrix}\begin{pmatrix}
I_1&S\\
0&I_2
\end{pmatrix}^{-1}
=
\begin{pmatrix}
A_1+F_0C_1 & 0\\
-F_2C_1&A_2+F_2C_1S
\end{pmatrix},
\end{array}
 \end{equation}
which is obviously Hurwitz.
This shows that the error system \dref{wxh20202061803} is  stable and hence \dref{wxh20202061808} holds true
in terms of  \dref{wxh20202061802}.
This completes the proof of the lemma.
\end{proof}

\section{Preliminaries on abstract linear systems}\label{preliminary}
In order to extend the finite-dimensional  results in Section \ref{motivation} to the
infinite-dimensional systems, we  introduce some necessary background on the infinite-dimensional linear systems, in particular for those
systems with unbounded control and observation operators,  which has been extensively discussed in \cite{TucsnakWeiss2009book}.

Suppose that $X$ is a Hilbert space and  $A :  D (A)\subset X \to X $ is a densely defined operator  with $\rho(A) \neq \emptyset$. The operator $A$ can determine two Hilbert spaces:
$( D (A), \|\cdot\|_1)$ and $([ D (A^*)]', \|\cdot\|_{-1})$, where
$ [ D (A^*)]' $ is the dual space of $  D (A) $ with respect to   the pivot space $X$, and the norms $\|\cdot\|_1$ and $\|\cdot\|_{-1}$ are defined   respectively by
\begin{equation} \label{20191141722}
 \left\{\begin{array}{l}
 \disp \|x\|_1=\|(\beta-A)x\|_X,\ \ \forall\ x\in  D (A),\crr
 \disp \|x\|_{-1}=\disp \|(\beta-A)^{-1}x\|_X,   \ \ \forall\ x\in X,
 \end{array}\right. \ \ \beta\in\rho(A).
\end{equation}
These two spaces are independent of the choice of $\beta\in\rho(A)$ because   different  choices
of $\beta$ lead  to  equivalent norms.
 For the sake of brevity, we denote the two spaces   as
$  D (A) $ and $ [ D (A^*)]' $ thereafter.
 The adjoint of $A^*\in \mathcal{L}( D (A^*),X)$, denoted by  $\tilde{A}$, is defined    as
 \begin{equation} \label{20191121602}
 \disp     \langle \tilde{A} x,y\rangle_{ [ D (A^*)]',  D (A^*)}=
 \langle x,A ^*y\rangle_{X },\ \ \forall\ x\in X ,\  y\in  D (A^*).
\end{equation}
Evidently,  $\tilde{A} x=Ax$ for any $x\in  D (A)$, which means that
  $\tilde{A}\in \mathcal{L}(X,  [ D (A^*)]')$ is an extension  of $A $. Since $A$ is
densely defined,  the extension is unique. By \cite[Proposition 2.10.3]{TucsnakWeiss2009book},
  $(\beta-\tilde{A})\in \mathcal{L}(X,[ D (A^*)]')$  and
$(\beta-\tilde{A})^{-1}\in \mathcal{L}( [ D (A^*)]',X)$, which imply that
$\beta-\tilde{A}$ is an isomorphism from $X$ to $[ D (A^*)]' $.
If the operator $A$ generates a $C_0$-semigroup $e^{At}$ on $X$, then so is for its extension $\tilde{A}$ and  $ e^{\tilde{A}t}=  (\beta-\tilde{A})  e^{At}(\beta-\tilde{A})^{-1}$ for any $ \beta\in\rho(A)$.

Suppose that $Y$ is an output Hilbert space and  $C\in \mathcal{L}(D(A),Y)$. The
$\Lambda$-extension of $C $  with respect to $A $  is defined by
   \begin{equation} \label{20206121528}
 \left\{ \begin{array}{l}
 \disp C_{ \Lambda}x=\lim\limits_{\lambda\rightarrow +\infty}
C \lambda(\lambda  - {A} )^{-1}x,\ \ x\in  D (C_{ \Lambda}),\crr
\disp D (C_{ \Lambda})=\{x \in X \ |\  \mbox{the above limit exists}\}.
\end{array}\right.
\end{equation}
  Define the norm
\begin{equation} \label{20205281609}
  \|x\|_{D(C_{\Lambda})}= \|x\|_X+\sup_{\lambda\geq\lambda_0}
  \|C \lambda(\lambda  - {A} )^{-1}x\|_Y,\ \ \forall\ x\in D(C_{\Lambda}),
\end{equation}
where $\lambda_0\in \R$ is any number so that  $[\lambda_0, \infty ) \subset \rho(A) $.
By  \cite[Proposition 5.3]{Weiss1994MCSS},  $D(C_{ \Lambda})$ with norm   $\|\cdot\|_{D(C_{\Lambda})}$ is
   a Banach space  and
$C_{ \Lambda}\in \mathcal{L}(  D(C_{ \Lambda}),Y)$.
Moreover,
  there exist continuous embeddings:
 \begin{equation} \label{20206121529}
 D(A) \hookrightarrow  D(C_{ \Lambda})\hookrightarrow  X   \hookrightarrow [ D (A ^*)]'.
\end{equation}

\begin{proposition}\label{Regular}  (\cite{Weiss1997TAC})
Let $X$, $U$ and $Y$ be the state space,     input  space and the output space, respectively.
The system
 $(A,B,C)$ is said to be a regular linear system    if and only if
the following assertions hold:

  (i)~ $A$ generates a $C_0$-semigroup $e^{At}$ on $X$;

 (ii)~ $B\in \mathcal{L}(U,[D(A^*)]')$ and  $C\in \mathcal{L}(D(A), Y)$ are admissible for $e^{At}$;

 (iii)~ $C_{\Lambda}(s -\tilde{A})^{-1}B$ exists for
some (hence for every) $s\in \rho(A)$;

(iv)~ $s\rightarrow \|C_{\Lambda}(s -\tilde{A})^{-1}B\|$ is bounded on some right half-plane.
  \end{proposition}
\begin{definition}\label{De20204191058} (\cite{Weiss1997TAC})
  Let $X$ and $Y$  be Hilbert spaces, let    $A$  be
the generator of $C_0$-semigroup $e^{At}$ on $X$ and let
  $C\in \mathcal{L}(D(A),Y)$. The operator  $L\in \mathcal{L}(Y,[D(A^*)]')$
is said to detect  system  $(A,C)$  exponentially
 if (a)~ $(A,L,C)$ is a regular linear system; (b)~  there exists an $s\in\rho(A)$ such that $I$   is an admissible feedback operator for    $C_{\Lambda}(s-\tilde{A})^{-1}L$; (c)~  $A+L C_{\Lambda}$ is exponentially stable.
\end{definition}

For other concepts of the
admissibility for both control and observation operators, and regular linear systems,
we refer to \cite{Weiss1989SICON,Weiss1989,Weiss1994MCSS,Weiss1994TAMS}.

\begin{definition}\label{WXH2020220943}\cite{FPart1}
Suppose that  $X$  is  a Hilbert space and
  $A_j: D(A_j)\subset X\to X$ is a densely defined operator with $\rho(A_j)\neq\emptyset$,
  $j=1,2$.
  We say that the operators   $ A_1 $  and  $ A_2 $ are  similar     with the
 transformation $P$, denoted by $ A_1 \sim_{P}  A_2 $,   if
  the  operator  $P \in \mathcal{L}(X)$  is   invertible   and satisfies
 \begin{equation} \label{wxh20202291136}
PA_1P^{-1}=A_2 \ \ \mbox{and}\ \   D (A_2) = P  D (A_1).
\end{equation}
\end{definition}

\begin{lemma}\label{THwxh20202291137}
Let $X $ and $Y$ be Hilbert spaces.
Suppose that  the operator $A_j :  D (A_j )\subset X \to X $  generates a $C_0$-semigroup
 $e^{A_jt}$ on   $X $ and $C_j\in \mathcal{L}( D (A_j),Y)$, $j=1,2$.
  If there is a $P\in \mathcal{L}(X)$ such that  $A_1\sim_{P} A_2$ and  $C_1=C_2P$,
 then, the following assertions hold true:

(i)~ $C_1$ is admissible for $e^{A_1 t}$ if and only if  $ C_2$
 is admissible for  $e^{A_2t}$;

 (ii)~   $(A_1,C_1)$ is exactly (approximately) observable if and only if   $(A_2, C_2)$ is exactly (approximately) observable.
\end{lemma}
\begin{proof}
Since  $A_1\sim_{P} A_2$, for any $x_2\in  D (A_2)$,  there exists an  $  x_1\in   D (A_1)$ such that $Px_1=x_2$.
Consequently,
\begin{equation}\label{wxh2020371534}
C_2e^{A_2t}x_2 = C_2 Pe^{A_1t}P^{-1}Px_1 = C_1e^{A_1t}x_1  .
\end{equation}
Since $\|x_2\|_{ X }\leq \|P\|\|x_1\|_{ X}$ and  $\|x_1\|_{X}\leq \|P^{-1}\|\|x_2\|_{X}$,
the desired  results can be concluded from the definitions of  the admissibility and the  exact  (approximate) observability.
\end{proof}

\section{Well-posedness of observation system  }\label{Se.4}
This section is devoted to the well-posedness of open-loop system \dref{wxh20202131701}.
We shall show that   the mapping from each initial data and control input
to the state and the output  is continuous.
System \dref{wxh20202131701} can be written as an abstract triple  $(\A,\mathcal{B}, \mathcal{C})$, where
 \begin{equation} \label{20205121037}
 \A =\begin{pmatrix}
 \tilde{A}_1&B_1C_{2\Lambda}\\0&\tilde{A}_2
\end{pmatrix},
\left.\begin{array}{l}
 \disp D(\A )=\left\{\begin{pmatrix}
  x_1 \\x_2 \end{pmatrix} \in X_{1}\times X_{2}  \ \Big{|} \
  \begin{array}{l}\disp \tilde{A}_1x_1 +B_1C_{2\Lambda}x_2\in X_1\\ \tilde{A}_2x_2\in X_2
  \end{array} \right\},
 \end{array}\right.
 \end{equation}
and
 \begin{equation} \label{2020521722}
 \mathcal{B}=
 \begin{pmatrix}
 0\\
 B_2
 \end{pmatrix},\ \ \  \mathcal{C}=( C_{1\Lambda },0),\ \ \ D(\mathcal{C})=D(C_{1\Lambda })\times X_2.
\end{equation}
   \begin{lemma}\label{Lm2020652112}
 Suppose that   the operator  $A_j$ generates  a $C_0$-semigroup $e^{A_jt}$ on $X_j$, $B_j\in \mathcal{L}(U_j,[D(A_j^*)]')$ is admissible for $e^{A_jt}$ and  $C_j \in \mathcal{L}(D(A_j),Y_j) $ is admissible for
 $e^{A_jt}$ with $Y_2=U_1$, $j=1,2$.
Then,   the operator $\A$    defined by \dref{20205121037} generates a
 $C_0$-semigroup $e^{\A  t}$ on $X_1\times X_2$.  Moreover,
 the following assertions hold:

 (i)  If  system $(A_1,B_1,C_{1 })$ is regular, then
  $\mathcal{C} $  is  admissible for $e^{\A t}$;

  (ii)  If   system $(A_2,B_2,C_{2 })$ is regular, then
   $\mathcal{B} $  is  admissible for $e^{\A t}$;

   (iii)  If both   systems $(A_1,B_1,C_{1})$  and $(A_2,B_2,C_{2 })$ are regular, then
  system $(\A,\mathcal{B},\mathcal{C})$ is  also regular.
 \end{lemma}
\begin{proof}
 It has been proved in previous study \cite{FPart1} that   the operator $\A$     generates a
 $C_0$-semigroup $e^{\A  t}$ on $X_1\times X_2$. We only need to prove (i),  (ii) and (iii).

Proof of (i).  Since    $C_2$  is admissible for $e^{A_2 t}$, for any $\tau>0$,  there exists a constant $c_{\tau}>0$
 such that
\begin{equation} \label{2020661410}
\int_0^\tau \left\| C_{2}e^{A_2s}x_2\right\|_{Y_2}^2ds\leq
c_{\tau}   \left\| x_2\right\|_{X_2}^2, \ \ \forall\
x_2\in   D(A_2).
\end{equation}
 Since  $(A_1,B_1,C_{1 })$ is a regular linear system,  there exists a constant $M_{\tau}>0$ such that
\begin{equation} \label{2020661213}
\int_0^\tau \left\| C_{1\Lambda }\int_0^te^{A_1(t-s)}B_1u_1(s)ds\right\|_{Y_1}^2dt\leq
M_{\tau} \int_0^\tau \left\|  u_1(s) \right\|_{Y_2}^2ds,\ \ \forall\
u_1\in  L^2([0,\tau]; Y_2).
\end{equation}
Combining \dref{2020661410} and \dref{2020661213}, we arrive at
\begin{equation} \label{2020661413}
 \int_0^\tau \left\| C_{1\Lambda }\int_0^te^{A_1(t-s)}B_1C_{2\Lambda}e^{A_2s}x_2ds\right\|_{Y_1}^2dt\leq
 M_{\tau}c_{\tau}\left\| x_2\right\|_{X_2}^2, \ \ \forall\
x_2\in   D(A_2).
\end{equation}
A straightforward computation shows that, for any  $(x_1,x_2)^\top\in D(\A)$,
\begin{equation} \label{2020661208}
\mathcal{C}e^{\A t}(x_1,x_2)^{\top}=C_{1\Lambda }e^{A_1t}x_1+C_{1\Lambda }\int_0^te^{A_1(t-s)}B_1C_{2\Lambda }e^{A_2s}x_2ds,
\end{equation}
which, together with \dref{2020661413} and
the admissibility of $C_1$ for $e^{A_1 t}$,  leads to
\begin{equation} \label{2020661417}
 \int_0^\tau \left\| \mathcal{C}e^{\A t}(x_1,x_2)^{\top}ds\right\|_{Y_1}^2dt\leq L_{\tau}\|(x_1,x_2) ^{\top}\|^2_{X_1\times X_2}, \ \
 \forall\ \ (x_1,x_2) ^{\top}\in D(\A),
\end{equation}
 where $L_{\tau}$ is a positive constant. Hence,
 $\mathcal{C} $  is  admissible for $e^{\A t}$.

 Proof of (ii). When  $(A_2,B_2,C_{2 })$ is a regular linear system, for any  $ u_2\in L_{\rm loc}^2([0,\infty); U_2)$ and $t\geq0$, it follows that
 \begin{equation} \label{2020661419}
 \int_0^t e^{A_2(t-s)}B_2u_2(s)ds\in X_2 \ \ \mbox{and}\ \
 C_{2\Lambda }\int_0^t e^{A_2(t-s)}B_2u_2(s)ds\in L_{\rm loc}^2([0,\infty); U_1).
\end{equation}
These, together with
the admissibility of $B_1$ for $e^{A_1t}$, lead  to
 \begin{equation} \label{2020661426}
 \int_0^t e^{\A(t-s)}\mathcal{B}u_2(s)ds =\begin{pmatrix}
   \int_0^t e^{A_1(t-\alpha)}B_1[C_{2\Lambda }\int_0^\alpha e^{A_2(\alpha-s)}B_2u_2(s)ds]
                                             d\alpha  \\
                                            \int_0^t e^{A_2(t-s)}B_2u_2(s)ds
                                          \end{pmatrix}\in X_1\times X_2.
\end{equation}
 Hence, $\mathcal{B} $  is  admissible for $e^{\A t}$.

 Proof of (iii).   Since both $(A_1,B_1,C_{1})$  and $(A_2,B_2,C_{2})$ are regular linear systems,
 for any   $\lambda\in  \rho(A_1)\cap \rho(A_2)\subset \rho(\A)$, we  conclude that
 $C_{j\Lambda}(\lambda-\tilde{A}_j)^{-1}B_j\in \mathcal{L}(U_j,Y_j)$ and $\lambda \rightarrow \|C_{j\Lambda}(\lambda-\tilde{A}_j)^{-1}B_j\|$ is bounded on some right half-plane, $j=1,2$. Moreover,
   a simple computation shows that
  \begin{equation} \label{2020661433}
 \begin{array}{ll}
 \disp\mathcal{C}   (\lambda- {\A})^{-1}\mathcal{B} &\disp =
 \mathcal{C} \left[\lambda-
  \begin{pmatrix}
           \tilde{A}_1 &B_1C_{2\Lambda}\\
           0&\tilde{A}_2
         \end{pmatrix}\right]^{-1}\mathcal{B}
 \crr
  &\disp =(C_{1\Lambda},0)
    \begin{pmatrix}
           (\lambda-\tilde{A}_1)^{-1} &(\lambda-\tilde{A}_1)^{-1}B_1C_{2\Lambda}(\lambda-\tilde{A}_2)^{-1}\\
           0&(\lambda-\tilde{A}_2)^{-1}
         \end{pmatrix}
    \begin{pmatrix}
    0 \\
    B_2
  \end{pmatrix}\crr
  &\disp=C_{1\Lambda}
     (\lambda-\tilde{A}_1)^{-1}B_1C_{2\Lambda}(\lambda-\tilde{A}_2)^{-1} B_2.
   \end{array}
  \end{equation}
     Consequently,
    $\mathcal{C}  (\lambda- {\A})^{-1}\mathcal{B} \in \mathcal{L}(U_2,Y_1)$ and
  $\lambda\to \|\mathcal{C} (\lambda-\tilde{\A})^{-1}\mathcal{B}\|$  is bounded on some right
half-plane. By Proposition \ref{Regular},   $(\A,\mathcal{B}, \mathcal{C})$ is a regular linear system.
\end{proof}

As an immediate consequence of  Lemma \ref{Lm2020652112}, we arrive at the following Theorem:
\begin{theorem}\label{Th20207212135}
Suppose that $(A_1,B_1,C_{1 })$  and $(A_2,B_2,C_{2})$  are    regular linear systems. Then,
system \dref{wxh20202131701} is  well-posed:
  For  any     $(x_1(0),x_2(0)) ^{\top}\in X_1\times X_2$  and   $u\in L^2_{\rm loc}([0,\infty);U_2)$, there exists a unique solution $(x_1,x_2)^{\top}\in C([0,+\infty);X_1\times X_2) $ to system \dref{wxh20202131701} such that
   \begin{equation} \label{20207212139}
  \|(x_1(t),x_2(t))^{\top}\|_{X_1\times X_2}+\int_0^t\|y(s)\|_{Y_1}^2ds\leq C_t
  \left[\|(x_1(0),x_2(0))^{\top}\|_{X_1\times X_2}+\int_0^t\|u(s)\|_{U_2}^2ds\right]
\end{equation}
   for any $t>0$, where $C_t>0$ is a  constant that is independent of
    $(x_1(0),x_2(0))  $  and   $u$.

\end{theorem}
\section{Observability of system \dref{wxh20202131701}}\label{Observability}
In this section, we consider the observability of system \dref{wxh20202131701}.
We denote by $\rho_{\infty}(A_1)$ the connected component of $\rho (A_1)$ which
contains some right half-plane.  This  set   has been used  in
\cite[Proposition 2.4.3, p.34]{TucsnakWeiss2009book}.
 Obviously, there is only one such component. In
particular, if $\sigma(A_1)$ is countable, which  is often the case in applications, then
$\rho_{\infty}(A_1)=\rho (A_1)$.
\begin{theorem}\label{Th2020681759}
Let $X_2$ and $Y_1 $ be   finite-dimensional Hilbert spaces and let
$(A_1,B_1,C_1)$  be a regular linear system with  the state space $X_1$,   input space  $U_1$ and the  output   space $Y_1 $.
Suppose that   $A_2\in\mathcal{L}(X_2)$,  $C_2 \in\mathcal{L}(X_2,U_1)$   and
\begin{equation} \label{wxh202078843}
  \sigma(A_2)\subset \rho_{\infty}(A_1) .
\end{equation}
Then,  system  $(\A ,\mathcal{C} )$ is exactly (approximately) observable if and only if system $(A_1,C_1)$ is exactly  (approximately)  observable and
\begin{equation}\label{wxh202079943}
 {\rm Ker} \left[C_{1\Lambda}(\lambda-\tilde{A}_1)^{-1}B_1C_2\right]\cap  {\rm Ker}(\lambda-A_2) =\{0\}, \ \ \forall\ \lambda\in\sigma(A_2),
\end{equation}
where the operators
  $\A$ and $\mathcal{C}$ are  given by \dref{20205121037} and \dref{2020521722}, respectively.
\end{theorem}
\begin{proof}
The assumption \dref{wxh202078843} implies that $\sigma(A_1)\cap \sigma(A_2)=\emptyset$.
 It follows from  \cite[Lemma 4.2]{FPart1}
that the Sylvester equation  ${A}_1P-P{A}_2 =B_1C_{2} $  admits a unique solution  $P\in \mathcal{L}(X_2,X_1)$  such
that $C_{1\Lambda}P\in\mathcal{L}(X_2,Y_1)$ and
\begin{equation} \label{wxh202078842}
  \tilde{A}_1Px_2-P {A}_2x_2 =B_1C_{2}x_2,\ \ \forall\ x_2\in X_2.
\end{equation}
 In terms of the solution $P$,
we introduce  an invertible  transformation
 $\mathbb{P}\in \mathcal{L}( X_1 \times X_2)$
 by
\begin{equation} \label{wxh202089913}
\left.\begin{array}{l}
\disp    \mathbb{P}  \left( {x}_1 , {x}_2 \right)^{\top}= \left( {x}_1+P  {x}_2,\ {x}_2 \right)^{\top},\ \ \forall\ ( {x}_1, {x}_2)^\top\in X_1\times X_2,
 \end{array}\right.
\end{equation}
whose  inverse is given by
\begin{equation} \label{wxh202089914}
\left.\begin{array}{l}
\disp    \mathbb{P}^{-1} \left( {x}_1 , {x}_2 \right)^{\top}= \left(   {x}_1- P  {x}_2, {x}_2\right)^{\top},\ \ \forall \left( {x}_1 , {x}_2 \right)^{\top}\in X_1\times X_2.
 \end{array}\right.
\end{equation}
Define $  \A _{\mathbb{P} }= {\rm diag}( {A}_1,A_2)$ with $D(\A _{\mathbb{P} })=
D(A_1)\times X_2$. Since $X_2$ is finite-dimensional,  we have
\begin{equation} \label{20208232011}
  D(\A )=\left\{(  x_1,x_2 )^{\top} \in X_{1}\times X_{2}  \ \Big{|} \
   \tilde{A}_1x_1 +B_1 C_{2} x_2\in X_1 \right\}.
    \end{equation}
 For any $(x_1,x_2)^{\top}\in D(\A_{\mathbb{P}})$, it follows from \dref{wxh202078842} and the fact $P\in \mathcal{L}(X_2,X_1)$ that
\begin{equation} \label{20208232015}
  \tilde{A}_1(x_1-Px_2) +B_1C_{2}x_2=
  \tilde{A}_1 x_1 -PA_2x_2  \in X_1,
  \end{equation}
which, together with \dref{20208232011} and \dref{wxh202089914}, leads to that
$\mathbb{P}^{-1}(x_1,x_2)^{\top}\in D(\A)$. Hence, $D(\A_{\mathbb{P}})\subset \mathbb{P}D(\A)$
due to the arbitrariness of $(x_1,x_2)$. On the other hand, for any $(x_1,x_2)^{\top}\in D(\A )$,
by \dref{wxh202078842}, \dref{20208232011} and
the fact $P\in \mathcal{L}(X_2,X_1)$, it follows that
\begin{equation} \label{20208232021}
\tilde{A}_1(x_1+Px_2)= \tilde{A}_1 x_1+B_1C_{2}x_2 +\tilde{A}_1Px_2 -B_1C_{2}x_2
 =[\tilde{A}_1 x_1+B_1C_{2}x_2]+PA_2x_2\in X_1,
\end{equation}
which implies that  $\mathbb{P}(x_1,x_2)^{\top}\in D(\A_{\mathbb{P}})$. Hence,
 $\mathbb{P}D(\A)\subset D(\A_{\mathbb{P}})$.
 As a result,   $\mathbb{P}D(\A)= D(\A_{\mathbb{P}})$. Moreover,
  a simple computation shows that
\begin{equation} \label{wxh202078747}
  {\A} \sim_{\mathbb{P}  }  {\A}_{\mathbb{P} }
 \ \ \mbox{and}\ \  \mathcal{C} \mathbb{P} ^{-1} =(C_{1\Lambda},  -C_{1\Lambda}P) .
\end{equation}
By Lemma \ref{THwxh20202291137},
$(\A,\mathcal{C}  ) $ is exactly  (approximately)  observable if and only if
$( {\A}_{\mathbb{P} },\mathcal{C} \mathbb{P} ^{-1} ) $ is exactly  (approximately)  observable.
Now, it  suffices to prove that
$( {\A}_{\mathbb{P} },\mathcal{C} \mathbb{P} ^{-1} )$ is exactly  (approximately)  observable if and  only if $(A_1,C_1)$ is exactly  (approximately)  observable and \dref{wxh202079943} holds.

Actually, for any $\lambda\in\sigma(A_2)$,
since $\lambda\notin\sigma(A_1)$,
it follows from \dref{wxh202078842} that
  \begin{equation} \label{2020691443}
P =(\lambda-\tilde{A}_1)^{-1}P(\lambda-A_2) -(\lambda-\tilde{A}_1)^{-1}B_1C_2.
\end{equation}
 Suppose that   $A_2x_2= \lambda x_2$, $x_2\in X_2$.  Then, \dref{2020691443} yields
 \begin{equation} \label{2020521758}
 -C_{1\Lambda}P  x_2=  C_{1\Lambda}(\lambda-\tilde{A}_1)^{-1}B_1C_2x_2.
\end{equation}
 By \cite[Remark 1.5.2, p.15]{TucsnakWeiss2009book}, $(A_2,  -C_{1\Lambda}P  ) $  is observable if and only if
 \begin{equation}\label{20208232033}
 {\rm Ker} \left(C_{1\Lambda}P \right)\cap  {\rm Ker}(\lambda-A_2) =\{0\}, \ \ \forall\ \lambda\in\sigma(A_2) .
\end{equation}
Combining  \dref{2020521758} and \dref{20208232033}, we conclude that  $(A_2,  -C_{1\Lambda}P  ) $  is observable if and only if  \dref{wxh202079943} holds.

When  $(\tilde{A}_1,C_{1\Lambda})$ is exactly  (approximately)  observable and \dref{wxh202079943} holds true,
then $(A_2,  -C_{1\Lambda}P  ) $ is observable.  Furthermore,
by  \dref{wxh202078843} and
\cite[Theorem 6.4.2, p.190]{TucsnakWeiss2009book}
(\cite[Proposition 6.4.5, p.192]{TucsnakWeiss2009book}),
$(  \A _{\mathbb{P}},\mathcal{C} \mathbb{P} ^{-1} ) $ is exactly  (approximately)  observable.
 Conversely,
if $(  \A_{\mathbb{P}  },\mathcal{C} \mathbb{P} ^{-1} ) $  is exactly  (approximately)  observable,
by the block-diagonal structure of $\A_{\mathbb{P}  }$ and \dref{wxh202078843},    it is easily to obtain that  $( {A}_1,C_{1 })$ is exactly  (approximately)  observable and $(A_2,   -C_{1\Lambda}P  ) $ is observable (see, e.g., \cite[Lemma 10.2]{FPart1}). In particular,   \dref{wxh202079943} holds.
The proof is complete.
\end{proof}

\begin{remark} \label{Re2020691451}
When system $(A_1,B_1,C_1)$ is single-input-single-output and system  $(A_2,C_2)$ is observable,
the assumption  \dref{wxh202079943} can be replaced by
 \begin{equation}\label{2020691458}
 C_{1\Lambda}(\lambda-\tilde{A}_1)^{-1}B_1 \neq 0,\ \ \lambda\in\sigma(A_2).
\end{equation}
Indeed,
suppose that
$x_2\in {\rm Ker} \left[C_{1\Lambda}(\lambda-\tilde{A}_1)^{-1}B_1C_2\right]\cap  {\rm Ker}(\lambda-A_2)$ for some  $\lambda\in\sigma(A_2)$.
By  \dref{2020691458} and the observability of $(A_2,C_2)$, we obtain  $C_2x_2=0$ and hence
  $x_2=0$. This yields  \dref{wxh202079943}.
   The condition \dref{2020691458} implies that every  point in $\sigma(A_2)$ is not the   transmission zero of system $(A_1,B_1,C_1)$.
\end{remark}

\section{Luenberger-like observer  for  abstract linear system}\label{Observer1}
In this section, we extend the results in  Section 2 from  finite-dimensional systems to     infinite dimensional ones.
Inspired by the finite-dimensional observer \dref{wxh20202061740}, an
infinite-dimensional Luenberger-like observer of \dref{wxh20202131701}  is   designed as
\begin{equation} \label{wxh202020617402}
 \left\{\begin{array}{l}
\disp \dot{\hat{x}}_1(t) =  A_1\hat{x}_1(t)+B_1C_{2 }\hat{x}_2(t)-F_1[y(t)- C_{1}\hat{x}_1(t)],\crr
\disp  \dot{\hat{x}}_2(t) =A_2\hat{x}_2(t)+F_2[y(t)-C_{1}\hat{x}_1(t)]+B_2u(t),
\end{array}\right.
\end{equation}
where the tuning  gain  operators  $F_1$ and $F_2$ are selected by the following scheme
\begin{itemize} \label{Scheme20207221434}
   \item Find $F_0\in \mathcal{L}(Y_1,[ D (A_1^*)]')$  to detect  system $(A_1, C_1)$  exponentially;

   \item Solve the following  Sylvester operator  equation:
   \begin{equation} \label{20204181538}
 (A_1+F_0C_{1\Lambda} ) S   - S   {A}_2 =   B_1C_{2 } ;
   \end{equation}

   \item Find $F_2\in \mathcal{L}(Y_1,[ D (A_2^*)]')$   to detect  system $(A_2, C_{1 \Lambda }S)$ exponentially;

   \item Set  $F_1=F_0+   {S}F_2$.

\end{itemize}

\begin{lemma}\label{wxh20206291110}
Suppose that $X_1$, $X_2$, $U_1$,   $Y_1$ are  Hilbert spaces and  $A_j: D(A_j)\subset X_j\to X_j$
 is a densely defined  operator with $\rho(A_j)\neq \emptyset$, $j=1,2$.
Suppose that $B_1\in \mathcal{L} ( U_1 , [D(A_1^*)]')$,
$C_1\in \mathcal{L} (D(A_1),Y_1)$, $C_2\in \mathcal{L} (D(A_2),U_1)$,
$F_0\in \mathcal{L}(Y_1,[ D (A_1^*)]')$,  $F_2\in \mathcal{L}(Y_1,[ D (A_2^*)]')$ and
$S  \in \mathcal{L}(X_2, X_1)$  solves the  Sylvester equation \dref{20204181538} in the sense that
 \begin{equation} \label{20207222100}
  (\tilde{A}_1+F_0C_{1\Lambda} )S x_2  - S   {A}_2x_2 =   B_1C_{2 }x_2,\ \ \forall \  x_2\in D(A_2).
\end{equation}
Then, the following assertions hold true:

(i)  If  $(A_1,B_1,C_1)$  is a regular linear system, then,  $C_{1\Lambda}S\in \mathcal{L}(D(A_2), Y_1) $;

(ii)  If the systems  $(A_2, F_2, C_{2} )$ and $(A_2, F_2, C_{1 \Lambda }S)$ are regular, then  there exists an extension of $S$, still denoted by $S$,  such that $ S F_2\in \mathcal{L}(Y_1, [D(A_1^*)]') $ and
\begin{equation} \label{20207222102}
   (\tilde{A}_1+F_0C_{1\Lambda} )  Sx_2 -  { S}   \tilde{A}_2x_2  =   B_1C_{2\Lambda} x_2 ,\ \ \forall\ \ x_2\in  X_{2F_2},
\end{equation}
where
\begin{equation}\label{20207222154}
 X_{2F_2}=D(A_2)+(\beta -\tilde{A}_2)^{-1}F_2U_2,\ \ \beta \in \rho(A_2) .
 \end{equation}
\end{lemma}
\begin{proof}
 (i)  For any  $x_2\in D(A_2)$, the assumption
 \dref{20207222100} yields $Sx_2\in D(C_{1\Lambda}) $  directly, provided
   $F_0\neq0$.
 When $F_0=0$,    it follows from  \dref{20207222100}
  that    $\tilde{A}_1 S x_2  -B_1C_{2 }x_2=S   {A}_2x_2 \in X_1 $, which  implies that $Sx_2\in D(A_1)+ (\alpha-\tilde{A}_1) ^{-1}B_1U_1$ with $\alpha\in \rho(A_1)$. Since
 $(A_1,B_1,C_1)$  is a regular linear system, $D(A_1)+ (\alpha-\tilde{A}_1) ^{-1}B_1U_1\subset D(C_{1\Lambda})$. Therefore, $Sx_2\in D(C_{1\Lambda})$ and
 $S(D(A_2))\subset  D(C_{1\Lambda})$.
  Since $S\in \mathcal{L}(D(A_2),X_1)$ and
  $D(C_{1\Lambda})$ is continuously embedded in $X_1$,
  the   inclusion $S(D(A_2))\subset  D(C_{1\Lambda})$
   implies that $S $    is a closed operator from $D(A_2)$  to $D(C_{1\Lambda})$.
 By the closed graph theorem, $S\in \mathcal{L}(D(A_2),D(C_{1\Lambda}))$
   and hence $C_{1\Lambda}S\in \mathcal{L}(D(A_2), Y_1) $.

(ii)  In terms of the solution  $S  \in \mathcal{L}(X_2, X_1)$  of  \dref{20207222100},
we define the operator $\tilde{S}$ by
   \begin{equation} \label{2020629855Ad722}
   \tilde{S}=   B_1C_{2\Lambda}({\beta }-\tilde{A}_2)^{-1}
   +({\beta }- \tilde{A}_1)S({\beta }-\tilde{A}_2)^{-1}-F_0C_{1\Lambda}S({\beta }-\tilde{A}_2)^{-1},\ \ \beta  \in \rho(A_2).
  \end{equation}
For  any $x_2\in X_2$,  since  $({\beta }-\tilde{A}_2)^{-1}x_2\in D(A_2)$, it follows from \dref{2020629855Ad722} and  \dref{20207222100}  that
\begin{equation} \label{2020629858Ad722}
 \begin{array}{ll}
 \disp   \tilde{S}x_2&\disp =   B_1C_{2\Lambda}({\beta }-\tilde{A}_2)^{-1}x_2+({\beta }- \tilde{A}_1)S({\beta }-\tilde{A}_2)^{-1}x_2-F_0C_{1\Lambda}S({\beta }-\tilde{A}_2)^{-1}x_2\crr
 &\disp =   -S  \tilde{A}_2({\beta }-\tilde{A}_2)^{-1}x_2+S{\beta } ({\beta }-\tilde{A}_2)^{-1}x_2 \crr
 &\disp =S({\beta }-\tilde{A}_2) ({\beta }-\tilde{A}_2)^{-1}x_2=Sx_2,
\end{array}
\end{equation}
which implies that $\tilde{S} $ is an extension of $S$.   On the other hand,     by \dref{2020629855Ad722} and the  regularity of      $(A_2,F_2,C_2)$ and  $(A_2, F_2,C_1S)$, we can conclude that
\begin{equation} \label{2020629924Ad722}
     \begin{array}{l}
 \disp   \tilde{S} F_2
 =B_1C_{2\Lambda}({\beta }-\tilde{A}_2)^{-1}F_2 +({\beta }-\tilde{A}_1) S({\beta }-\tilde{A}_2)^{-1}F_2 -F_0C_{1\Lambda}S({\beta }-\tilde{A}_2)^{-1}F_2,
 \end{array}
 \end{equation}
which implies that $\tilde{S}F_2\in \mathcal{L}(Y_1, [D(A_1^*)]')$.
Moreover, for any $y_1\in Y_1$, if we let $x_{{\beta }}=({\beta }-\tilde{A}_2)^{-1}F_2y_1$, then, it follows from \dref{2020629924Ad722} and \dref{2020629858Ad722} that
\begin{equation} \label{20205281647Ad722}
         \begin{array}{l}
\disp (\tilde{A}_1+F_0C_{1\Lambda} )  \tilde{S} x_{{\beta }}  -
 B_1C_{2\Lambda}x_{{\beta }} = {\beta }\tilde{S} x_{{\beta }}-\tilde{S}F_2y_1
={\beta }\tilde{S} x_{{\beta }}-\tilde{S} ({\beta }-\tilde{A}_2) x_{{\beta }} =\tilde{S}\tilde{A}_2x_{{\beta }},
\end{array}
\end{equation}
which implies that $\tilde{S}$  solves  the Sylvester equation \dref{20204181538} on $({\beta }-\tilde{A}_2)^{-1}F_2Y_1$.
Since  $\tilde{S}|_{X_2}=S$, \dref{20207222100}  and \dref{20207222154}, we can obtain  \dref{20207222102} easily with the replacement of   $S$ by  $\tilde{ S}$.
\end{proof}
\begin{theorem}\label{wxhTh20202211035}
Let  $(A_j,B_j,C_j)$  be a regular linear system with the state space $X_j$,  input  space $U_j$ and the output space $Y_j$,  $j=1, 2$.
Suppose that $U_1=Y_2$,  $F_0\in \mathcal{L}(Y_1,[ D (A_1^*)]')$ detects system $(A_1, C_1)$ exponentially, $F_2\in \mathcal{L}(Y_1,[ D (A_2^*)]')$  detects system $(A_2, C_{1 \Lambda }S)$  exponentially, $(A_2,F_2,C_2)$ is a regular linear system  and $F_1=F_0+   {S}F_2$, where
$S\in \mathcal{L}(X_2,X_1)$  is the solution  of    Sylvester equation \dref{20204181538}
in the sense of  \dref{20207222100}.
Then,   the observer \dref{wxh202020617402}  of system \dref{wxh20202131701} is well-posed:
For any $(\hat{x}_1(0),\hat{x}_2(0))^\top\in X_1\times X_2$ and $u\in L^2_{\rm loc}([0,\infty);U_2)$,
the observer \dref{wxh202020617402}    admits a unique solution $(\hat{x}_1,\hat{x}_2)^\top\in C([0,\infty);X_1\times X_2)$  such that
\begin{equation} \label{wxh202020618082}
e^{\omega t} \|(x_1(t)-\hat{x}_1(t), x_2(t)-\hat{x}_2(t))^\top\|_{X_1\times X_2}\to  0 \ \ \mbox{as}\ \ t\to\infty,
\end{equation}
where $\omega$  is a positive constant that is independent of $t$.
\end{theorem}
\begin{proof}
By Theorem \ref{Th20207212135},    for any
$(  {x}_1(0),  {x}_2(0) )^\top\in  X_1\times X_2 $ and $u \in L^2_{\rm loc}([0,\infty);U_2)$,
system \dref{wxh20202131701}    admits a unique solution $( x_1,x_2 )^\top\in C([0,\infty); X_1\times X_2 )$ such that $y= {C}_{1\Lambda}x_1\in  L^2_{\rm loc}([0,\infty);Y_1)$.
Set the errors
\begin{equation} \label{2020661506}
\tilde{x}_i(t)=x_i(t)-\hat{x}_i(t),\ \ i=1,2,
\end{equation}
which  are  governed by
\begin{equation} \label{2020661504}
 \left\{\begin{array}{l}
\disp \dot{\tilde{x}}_1(t) = (A_1+F_1C_{1 })\tilde{x}_1(t)+B_1C_{2 } \tilde{x}_2(t),\crr
\disp  \dot{\tilde{x}}_2(t) =  A_2 \tilde{x}_2(t)-F_2C_{1 }\tilde{x}_1(t).
\end{array}\right.
\end{equation}
System \dref{2020661504} can be written as
\begin{equation} \label{2020660515}
\frac{d}{dt}(\tilde{x}_1(t),\tilde{x}_2(t))^{\top}
= \mathscr{A}(\tilde{x}_1(t),\tilde{x}_2(t))^{\top} ,
\end{equation}
where
\begin{equation} \label{2020661507}
  \left\{\begin{array}{l}
 \disp  {\mathscr{A}}= \begin{pmatrix}
 \tilde{A}_1+F_1C_{1\Lambda}  &B_1C_{2\Lambda} \\
-F_2C_{1\Lambda} &  \tilde{A}_2
\end{pmatrix},\crr
\disp  D( {\mathscr{A}})= \left \{ \begin{pmatrix}
                                     x _1   \\
                                       x _2
                                   \end{pmatrix}  \in X_1\times X_2 \ \Big{|}\
                                   \begin{array}{l}
\disp (\tilde{A}_1+F_1C_{1\Lambda})  x _1+B_1C_{2\Lambda} x _2 \in X_1\\ \disp\tilde{A}_2 x _2-F_2C_{1\Lambda} x _1 \in X_2\end{array}\right \} .
\end{array}\right.
\end{equation}
In terms of the solution $S$ of   the Sylvester equation \dref{20204181538},
we introduce  a  bounded invertible  transformation $\mathbb{S}\in \mathcal{L}( X_1\times X_2)$  by
\begin{equation} \label{wxh20206291210}
\left.\begin{array}{l}
\disp    \mathbb{S}  \left( {x}_1 , {x}_2 \right)^{\top}= \left( {x}_1+S  {x}_2,\ {x}_2 \right)^{\top},\ \ \forall\ ( {x}_1, {x}_2)^\top\in X_1\times X_2,
 \end{array}\right.
\end{equation}
whose inverse is given by
\begin{equation} \label{wxh20206291211}
\left.\begin{array}{l}
\disp    \mathbb{S}^{-1} \left( {x}_1 , {x}_2 \right)^{\top}= \left(   {x}_1- S  {x}_2, {x}_2\right)^{\top},\ \ \forall \left( {x}_1 , {x}_2 \right)^{\top}\in X_1\times X_2.
 \end{array}\right.
\end{equation}
By virtue of Lemma \ref{wxh20206291110},    $C_{1\Lambda}S\in \mathcal{L}(D(A_2),Y_1)$,   $SF_2\in \mathcal{L}(Y_1,[ D (A_1^*)]')$ and \dref{20207222102} holds.
Since  system $(A_2,F_2,C_{1\Lambda}S)$ is regular, it follows that
 $C_{1\Lambda}S(\beta-\tilde{A}_2)^{-1}F_2\in \mathcal{L}(Y_1)$ for any $\beta\in \rho(A_2)$,
 which, together with \dref{20207222154} and the fact  $C_{1\Lambda}S\in \mathcal{L}(D(A_2),Y_1)$, implies that  $C_{1\Lambda}S (X_{2F_2})\subset  Y_1$.

 Define
\begin{equation} \label{2020429957AD510}
\mathscr{A}_\mathbb{S}= \begin{pmatrix}
\tilde{A}_1+F_0C_{1\Lambda} & 0\\
-F_2C_{1\Lambda}& \tilde{A}_2+F_2C_{1\Lambda}S
\end{pmatrix}
 \end{equation}
 with
  \begin{equation} \label{2020429956AD510}
\disp  D(\mathscr{A}_\mathbb{S})  =
\left\{\begin{pmatrix}
x_1\\x_2\end{pmatrix} \in X_1 \times X_{2}\ \Big{|}\
\begin{array}{l}
 \disp (\tilde{A}_1+F_0C_{1\Lambda})x_1\in X_1\\
\disp (\tilde{A}_2+F_2C_{1\Lambda}S)x_2-F_2C_{1\Lambda}x_1 \in X_2
\end{array}
  \right\}.
 \end{equation}
We claim that  $\mathscr{A} \sim_{\mathbb{S}} \mathscr{A}_\mathbb{S}$. Indeed, for any $( x _1, x _2)^\top\in D(\mathscr{A})$,
since $X_{2F_2}$ defined by \dref{20207222154} can be
  characterized as  (\cite[Remark 7.3]{Weiss1994MCSS}):
\begin{equation} \label{2020723853}
 X_{2F_2}=\left\{ x_2\in X_2\ | \  \tilde{A}_2x_2+F_2y_1  \in X_2, \ y_1 \in Y_1 \right\},
\end{equation}
$( x _1, x _2)^\top\in D(\mathscr{A})$ implies that
 $ x _2\in X_{2F_2}$ and hence $C_{1\Lambda}S x_2\in  Y_1$.
 By \dref{2020661507}, we have
\begin{equation}\label{wxh20206301652}
(\tilde{A}_2+F_2 C_{1\Lambda}S ) x _2-F_2C_{1\Lambda}( x _1+S x _2) = \tilde{A}_2 x _2-F_2C_{1\Lambda} x _1\in X_2.
\end{equation}
Noting that $F_1=F_0+SF_2$, it follows from  \dref{2020661507} that
\begin{equation}\label{wxh20206301231}
(\tilde{A}_1+F_0C_{1\Lambda} )x _1+SF_2 C_{1\Lambda} x _1+B_1C_{2\Lambda} x _2 =
 (\tilde{A}_1+F_1C_{1\Lambda})  x _1+B_1C_{2\Lambda} x _2 \in X_1,
\end{equation}
which, together with  \dref{20207222102}, \dref{wxh20206301652}, \dref{wxh20206301231}  and  $S\in \mathcal{L}(X_2,X_1)$, yields
\begin{equation}\label{wxh20206301232}
\begin{array}{ll}
\disp (\tilde{A}_1+F_0C_{1\Lambda} )( x _1+S x _2)&\disp
=(\tilde{A}_1+F_0C_{1\Lambda} ) x _1
 +B_1C_{2\Lambda} x _2+
S  \tilde{A}_2 x _2 \crr
&\disp \hspace{-3cm} =[(\tilde{A}_1+F_0C_{1\Lambda} ) x _1
+SF_2C_{1\Lambda} x _1+B_1C_{2\Lambda} x _2]+
S( \tilde{A}_2 x _2-F_2C_{1\Lambda} x _1) \in X_1.
\end{array}
\end{equation}
From \dref{2020429956AD510}, \dref{wxh20206301652} and \dref{wxh20206301232}, we can conclude that   $\mathbb{S}( x _1, x _2)^\top\in D(\mathscr{A}_{\mathbb{S}})$ and thus $\mathbb{S}(D(\mathscr{A}))\subset D(\mathscr{A}_{\mathbb{S}})$.

On the other hand, for any $(x_1,x_2)^\top\in D(\mathscr{A}_{\mathbb{S}})$, we have
 \begin{equation}\label{20207231012}
\tilde{A}_2 {x}_2-F_2C_{1\Lambda}( {x}_1-S {x_2}) \in X_2,
\end{equation}
which implies that
  $x_2\in X_{2F_2}$.
By $F_1=F_0+SF_2$,  \dref{20207222102}, \dref{2020429956AD510}  and  $S\in \mathcal{L}(X_2,X_1)$,
\begin{equation}\label{wxh20206301233}
(\tilde{A}_1+F_1C_{1\Lambda})({x}_1-Sx_2)+B_1C_{2\Lambda}x_2= (\tilde{A}_1+F_0C_{1\Lambda} )x_1
-S[ \tilde{A}_2x_2-F_2C_{1\Lambda}(x_1-Sx_2)] \in X_1.
\end{equation}
 It follows from \dref{20207231012} and \dref{wxh20206301233}  that  $\mathbb{S}^{-1}( {x}_1, {x}_2)^\top\in D(\mathscr{A})$ and thus $D(\mathscr{A}_{\mathbb{S}})\subset \mathbb{S}(D(\mathscr{A}))$.
 We therefore arrive at
 $D(\mathscr{A}_{\mathbb{S}})= \mathbb{S}(D(\mathscr{A}))$.
 By exploiting  \dref{20207222102}, a straightforward computation shows that $\mathbb{S}\mathscr{A}\mathbb{S}^{-1} (x_1,x_2)^\top=\mathscr{A}_{\mathbb{S}}(x_1,x_2)^\top$ for any $(x_1,x_2)^\top\in D(\mathscr{A}_{\mathbb{S}})$.
 This proves that $\mathscr{A}_{\mathbb{S}}$ and $\mathscr{A}$  are similar each other.

 Since $F_0 $  and $F_2\ $ detect exponentially  systems  $(A_1, C_1)$ and
 $(A_2, C_{1 \Lambda }S)$, respectively,  $e^{  (\tilde{A}_1+F_0C_{1\Lambda  })  t}$ and
 $e^{(\tilde{A}_2+F_2 C_{1\Lambda}S)t }$ are exponentially stable in $X_1$ and $ X_2$, respectively.
Moreover,  $C_1$ is admissible for $e^{  (\tilde{A}_1+F_0C_{1\Lambda  })  t}$  and
$F_2$ is admissible for $e^{(\tilde{A}_2+F_2 C_{1\Lambda}S)t }$.
 By   \cite[Lemma 3.3]{FPart1},
$\mathscr{A}_\mathbb{S}$  generates an exponentially stable   $C_0$-semigroup $e^{\mathscr{A}_\mathbb{S} t}$  on $X_1\times X_2$.
From the similarity of $\mathscr{A}_\mathbb{S}$
and $\mathscr{A}$, the operator $\mathscr{A}$ generates an
exponentially stable $C_0$-semigroup $e^{\mathscr{A}t}$ on
$X_1\times X_2$ as well.
As a result, system \dref{2020661504} with the initial state
$(\tilde{x}_1(0),\tilde{x}_2(0))= (x_1(0)-\hat{x}_1(0),x_2(0)-\hat{x}_2(0))$
 admits a unique solution $(\tilde{x}_1,\tilde{x}_2)^\top\in C([0,\infty);X_1\times X_2)$.
 Let $(\hat{x}_1(t),\hat{x}_2(t))= (x_1(t)-\tilde{x}_1(t),x_2(t)-\tilde{x}_2(t))$.
 A straightforward computation then shows that such a defined
$(\hat{x}_1(t),\hat{x}_2(t))$ solves the observer  \dref{wxh202020617402}  and satisfies \dref{wxh202020618082}. The uniqueness of the solution can be  obtained easily by the linearity of the observer.
The   proof is complete.
\end{proof}

\begin{remark}\label{Re2020813}
When a system is given,
one needs  to solve the Sylvester equation \dref{20204181538} to get the tuning gain operators $F_1$ and $F_2$   of the observer \dref{wxh202020617402}.
This is usually  not a trivial task, particularly for the case that both $A_1$ and $A_2$ are unbounded.
Consequently, Theorem \ref{wxhTh20202211035} does  not mean   that we  can always design  an   observer for the given system.
However, under some reasonable additional assumptions, we still can solve the
Sylvester equation  analytically or numerically even for  the cascade system involving  a multi-dimensional PDE
(see, e.g.,  \cite{Natarajan2016TAC} and \cite{Lassi2014SIAM}).
In particular,  the situation will become easier
provided one of $A_1$ and  $A_2$ is  bounded.
In this case, the solution of the   Sylvester equation \dref{20204181538} always exists  provided (see, e.g.,\cite{FPart1})
\begin{equation} \label{202068951**}
 \sigma({A}_1+F_0C_{1\Lambda})\cap \sigma(A_2)=\emptyset.
 \end{equation}
When the cascade system consists of an ODEs and a one-dimensional PDE, the problem becomes quite easy. In this case, an implementable way to solve the    Sylvester equation is given at the end of Section 5 of  the first paper \cite{FPart1} of this series works.
To show the effectiveness, this method will be applied to
observer design for ODEs  with output  delay and unstable heat equation with ODE sensor dynamics in Sections \ref{ODE+delay} and \ref{heat+ODE}, respectively.
\end{remark}

When $X_2$ is finite-dimensional, we can  characterize the existence of   the tuning gains $F_1$ and $F_2$ through the system \dref{wxh20202131701} itself.

\begin{corollary} \label{Co2020425849}
Let    $A_2 \in  \mathcal{L}(X_2) $ be a matrix with $\sigma(A_2)\subset\{ s\ |\ {\rm Re}s\geq0\}$,
and let    $(A_1,B_1, C_1)$ be a    regular linear system.   Suppose that
system $(A_1, C_1)$ is exponentially detectable
and $( \A , \mathcal{C})$  is approximately observable,
where $\A$ and $\mathcal{C}$ are given by \dref{20205121037} and \dref{2020521722}, respectively.
Then, there exist  $F_1\in \mathcal{L}(Y_1,[D(A^*_1)]') $ and
$F_2\in \mathcal{L}(Y_1,X_2) $ such that
the observer \dref{wxh202020617402}  of system \dref{wxh20202131701} is well-posed:
For any $(\hat{x}_1(0),\hat{x}_2(0))^{\top}\in X_1\times X_2$ and $u\in L^2_{\rm loc}([0,\infty);U_2)$,
the observer \dref{wxh202020617402}    admits a unique solution $(\hat{x}_1,\hat{x}_2)^{\top}\in C([0,\infty);X_1\times X_2)$  such that
\dref{wxh202020618082} holds for some  positive constant $\omega$.
\end{corollary}
\begin{proof}
Since system $(A_1, C_1)$ is exponentially detectable, there exists an $F_0\in \mathcal{L}(Y_1,[ D (A_1^*)]')$  to detect  system $(A_1, C_1)$ exponentially.
In particular, ${A}_1+F_0C_{1\Lambda  }$ generates an exponentially stable $C_0$-semigroup
$e^{({A}_1+F_0C_{1\Lambda  })t}$ on $X_1$.
Noting  that the  matrix $A_2$  satisfies $\sigma(A_2)\subset\{ s\ |\ {\rm Re}s\geq0\}$, we have
 \dref{202068951**}.
It follows from \cite{Natarajan2016TAC} that the Sylvester equation \dref{20204181538} admits a solution  $S\in \mathcal{L}(X_2,X_1)$  in the sense that
\begin{equation} \label{20207231431}
 (\tilde{A}_1+F_0C_{1\Lambda} )S x_2  - S   {A}_2x_2 =   B_1C_{2 }x_2,\ \ x_2\in X_2 .
   \end{equation}
By  Lemma   \ref{wxh20206291110}, we have  $C_{1\Lambda}S\in \mathcal{L}(X_2,Y_1)$.
Define
\begin{equation} \label{20207231447}
 \A_{F_0} =\begin{pmatrix}
 \tilde{A}_1+F_0C_{1\Lambda}&B_1C_{2\Lambda}\\0& {A}_2
\end{pmatrix}
\end{equation}
with
\begin{equation} \label{20207231447domain}
 \disp D(\A_{F_0} )=\left\{\begin{pmatrix}
  x_1 \\x_2 \end{pmatrix} \in X_{1}\times X_{2}  \ \Big{|} \
  \begin{array}{l}\disp (\tilde{A}_1+F_0C_{1\Lambda})x_1 +B_1C_{2\Lambda}x_2\in X_1\\ {A}_2x_2\in X_2
  \end{array} \right\}.
 \end{equation}
As in the proof of Theorem \ref{Th2020681759}, it follows from \dref{20207231431} that
\begin{equation} \label{20208122114}
\A_{F_0}\sim_{\mathbb{S}}  {\A}_\mathbb{S}  \ \ \mbox{and}\ \
  \mathcal{C}_\mathbb{S}=\mathcal{C} \mathbb{S}^{-1}= ( C_{1 \Lambda } , - C_{1\Lambda}S),
  \end{equation}
where  the operator  $\mathbb{S}$ is  given by      \dref{wxh20206291210}   and
${\A}_\mathbb{S}
= {\rm diag}(\tilde{A}_1+F_0C_{1\Lambda}, {A}_2) $ with
 \begin{equation} \label{20208241023}
  D(\A_\mathbb{S}) =\left\{(x_1,x_2)^{\top}
   \in\   X_1\times X_2\ |\
 (\tilde{A}_1+F_0C_{1\Lambda})x_1\in X_1,\ \ {A}_2x_2\in X_2\right\}.
\end{equation}
By a simple computation, the operator $\A_{F_0}$ can be written as
$\A_{F_0}=\A+\mathcal{F}_0\mathcal{C}$, where $\mathcal{F}_0=(F_0,0)^{\top}$.
Moreover, for any $u\in L_{\rm loc}^2([0,\infty); Y_1)$, it follows that
\begin{equation} \label{20208122058}
 \int_0^t e^{\A(t-s)}\mathcal{F}_0u (s)ds =
 \begin{pmatrix}
   \int_0^t e^{A_1(t-s)}F_0 u (s)ds  \\
    0
\end{pmatrix}\in X_1\times X_2,
\end{equation}
which implies that $\mathcal{F}_0$  is admissible for $e^{\A t}$.
By Lemma \ref{Lm2020652112},   $\mathcal{C}$ is  admissible for $e^{\A t}$ as well.
Noting that  $F_0\in \mathcal{L}(Y_1,[ D (A_1^*)]')$ detects system $(A_1, C_1)$ exponentially
and
\begin{equation} \label{wxh2020891330}
\mathcal{C}(s-\A)^{-1}\mathcal{F}_0= C_{1\Lambda}(s-\tilde{A}_1)^{-1}F_0,\ \ \forall\ s\in \rho(\A),
\end{equation}
we conclude that   $( \A ,\mathcal{F}_0, \mathcal{C})$ is   a regular linear  system and $I$   is an admissible feedback operator for    $\mathcal{C}(s-\A)^{-1}\mathcal{F}_0 $.
Since $( \A , \mathcal{C})$  is approximately observable,  it follows from
\cite[Remark 6.5]{Weiss1994MCSS} that system
$( \A +\mathcal{F}_0\mathcal{C} ,  \mathcal{C})=     ( \A_{F_0} ,  \mathcal{C})    $  is approximately   observable.
By  Lemma \ref{THwxh20202291137} and the similarity \dref{20208122114},  system  $(\A_{\mathbb{S}},  \mathcal{C}_\mathbb{S})$ is approximately observable as well.
Thanks to  \cite[Remark 6.1.8, p.175]{TucsnakWeiss2009book},  we have
\begin{equation} \label{2020723122}
{\rm Ker}(\lambda-  \A_\mathbb{S} )\cap {\rm Ker}\left( \mathcal{C}  \mathbb{S}^{-1}\right)=\{0\},\ \ \forall\ \lambda\in \sigma(\A_\mathbb{S})\subset\sigma(\tilde{A}_1+F_0C_{1\Lambda})\cup\sigma(A_2).
\end{equation}
For  any $h_2\in {\rm Ker}(\lambda-A_2)\cap {\rm Ker}(  C_{1\Lambda}S)$, it follows that
\begin{equation} \label{2020427929}
(0,h_2)^{\top}\in {\rm Ker}(\lambda-   \A_{\mathbb{S}})\cap {\rm Ker}\left( \mathcal{C}  \mathbb{S}^{-1}\right),\ \ \lambda\in \sigma(A_2),
\end{equation}
which, together with \dref{2020723122}, implies that $  {\rm Ker}(\lambda-A_2)\cap {\rm Ker}(  C_{1\Lambda}S)=\{0\}$.
Furthermore, by \cite[Remark 1.5.2, p.15]{TucsnakWeiss2009book}, $(A_2 , - C_{1\Lambda}S)$ is  observable.
From the pole assignment theorem of the  linear systems,  there exists an $F_2\in \mathcal{L}(Y_1,X_2)$ such that $A_2+F_2C_{1\Lambda}S$ is Hurwitz. Let   $F_1=F_0 +{S}F_2$.
The proof is then accomplished by    Theorem \ref{wxhTh20202211035}.
\end{proof}

The Corollary \ref{Co2020425849} covers   a fair amount of    dynamics compensations for
the finite-dimensional systems with infinite-dimensional sensor dynamics
dominated by the first order hyperbolic equation \cite{SmyshlyaevKrstic2008SCL},
heat equation \cite{Krstic2009SCL,TangSCL2011} as well as the one-dimensional wave equation \cite{Krstic2009TAC}.
In particular, the sensor delay for ODEs  in \cite{KrsticDelaybook} can also be compensated by
Corollary \ref{Co2020425849}.
\section{Disturbance compensation in output regulation}\label{secDisturbance}
Thanks to the abstract setting \dref{wxh20202131701}, it is found that the sensor dynamics compensation  of   linear  systems is closely related to the
disturbance and state estimation in   the output regulation, which were usually  considered as
two  different problems in   existing literatures. We refer the results about  sensor   dynamics compensation to  \cite{SmyshlyaevKrstic2008SCL,Krstic2009SCL,Krstic2009TAC,KrsticDelaybook} and the
results on output regulation to
\cite{Natarajan2016TAC,Lassi2017SIAM,Lassi2014SIAM,Lassi2010SIAM}.

 Consider the following  output regulation problem     in the state space $X_1$,   input  space $U_1$ and the output space $Y_1$:
  \begin{equation} \label{2020681703}
 \left\{\begin{array}{l}
\disp \dot{z}_1(t) = A_1 z_1(t)+B_d d(t)+B_1u(t) ,\crr
\disp y(t)=C_{1 }z_1(t)+r(t),
  \end{array}\right.
\end{equation}
 where $A_1:X_1\to X_1$, $B_1:U_1\to X_1$,  $C_1:X_1\to Y_1$ are the system operator, control  operator and the
 observation    operator, respectively,  $u(t)$ is the control, $d(t)$ is the disturbance in a Hilbert space $U_d$, $B_d:U_d\to X_1$,    $y(t)$ is the tracking  error, and $r(t)$ is the reference signal.

As in the conventional  output regulation problem discussed    in \cite{Lassi2014SIAM},  we suppose that the    disturbance and  reference signal  are generated from the following
exosystem in Hilbert space $X_2$:
   \begin{equation} \label{2020681715}
  \dot{z}_2(t) = A_2 z_2(t), \ d(t)=C_dz_2(t),\ \ r(t)= C_2z_2(t),
   \end{equation}
 where  $A_2\in \mathcal{L}(X_2)$,  $C_d:X_2\to U_d$ and   $C_2:X_2\to Y_1$.
In this section, we assume that
  \begin{equation} \label{2020681738}
  X_2=\mathbb{C}^n,\ C_d\in \mathcal{L}(X_2,U_d),\ C_2\in \mathcal{L}(X_2,Y_1), \ \sigma(A_2)\subset\{\lambda\ |\ {\rm Re} \lambda\geq0\}
   \end{equation}
  and $A_1$ generates an exponentially stable $C_0$-semigroup $e^{A_1t}$ on $X_1$.  Combining system \dref{2020681703} and exosystem  \dref{2020681715}, we obtain the
   following coupled system:
  \begin{equation} \label{20204271052}
 \left\{\begin{array}{l}
\disp \dot{z}_1(t) = A_1 z_1(t)+B_dC_dz_2(t)+B_1u(t) ,\crr
\disp \dot{z}_2(t) = A_2 z_2(t)  ,\crr
\disp y(t)=C_{1 }z_1(t)+C_2z_2(t).
\end{array}\right.
\end{equation}
 The main objective of the output regulation is   to design a   tracking   error based   feedback control  to regulate the
  tracking error to zero as $t\to\infty$.
 By  \cite{Natarajan2016TAC}, there exists a full state feedback law
$u(t)= -Qz_2(t)$
that solves the regulation  problem \dref{20204271052} if and only if the
  following regulator
equations
\begin{equation} \label{20205201035}
 \left\{\begin{array}{l}
\disp  {A}_1\Pi - \Pi A_2=-B_dC_d+B_1 Q,\crr
\disp C_{1\Lambda} \Pi+C_2=0
\end{array}\right.
\end{equation}
admits a solution  $\Pi \in \mathcal{L}(X_2,X_1)$
and $Q\in\mathcal{L}(X_2,U_1)$.
In order to improve the results in \cite{Natarajan2016TAC} where only the full state feedback is considered, we  consider, in this section, an observer based error feedback design.
By the   separation principle of the  linear systems, an    error based feedback can  be   designed  easily   provided that we  have an   observer for system \dref{20204271052}.
The only measurement for the observer design  is the tracking error $y(t)$.

Since system \dref{20204271052} without control is a  cascade of the control plant and
the exosystem, the observer design of system \dref{20204271052} is closely related to the
problem of sensor dynamics compensation.
Actually, after an invertible transformation, we can regard the control plant as the sensor dynamics of the exosystem.
Suppose that $(z_1,z_2) ^{\top}\in C([0,+\infty);X_1\times X_2)$ is a solution  of system \dref{20204271052}.
If we define, in terms of the solution of regulator  equations \dref{20205201035},  the transformation
\begin{equation} \label{20205201048}
\begin{pmatrix}
   x_1(t) \\
   x_2(t)
 \end{pmatrix}=
  \begin{pmatrix}
   I_1& -\Pi \\
   0&I_2
 \end{pmatrix} \begin{pmatrix}
   z_1(t) \\
   z_2(t)
 \end{pmatrix},
\end{equation}
 then system \dref{20204271052} is transferred into
\begin{equation} \label{20205201040}
 \left\{\begin{array}{l}
\disp \dot{x}_1(t) = A_1 x_1(t)+B_1Qx_2(t)+B_1u(t) ,\crr
\disp \dot{x}_2(t) = A_2 x_2(t)  ,\crr
\disp y(t)=C_{1 }x_1(t) ,
\end{array}\right.
\end{equation}
which takes   the same form as   system \dref{wxh20202131701}.
As a result,  an observer of system \dref{20204271052} can be designed by combining the transformation \dref{20205201048} and   Theorem  \ref{wxhTh20202211035}, which takes the form:
 \begin{equation} \label{20205201013}
 \left\{\begin{array}{l}
\disp \dot{\hat{z}}_1(t) =  A_1\hat{z}_1(t)+B_dC_{d}\hat{z}_2(t)
-K_1[y(t)- C_{1 }\hat{z}_1(t)-  C_{2 }\hat{z}_2(t)]+ B_1u(t),\crr
\disp  \dot{\hat{z}}_2(t) =A_2\hat{z}_2(t)+K_2[y(t)-C_{1 }\hat{z}_1(t)
-  C_{2 }\hat{z}_2(t)],
\end{array}\right.
\end{equation}
where the tuning  gains     $K_1\in \mathcal{L}(Y_1,X_1) $ and $K_2\in \mathcal{L}(Y_1,X_2)$
can be selected by the following scheme:
 \begin{itemize} \label{Scheme20207231653}
 \item Solve  the following  Sylvester equation on $X_2$:
\begin{equation} \label{wxh202078226}
{A}_1  \Gamma   - \Gamma   {A}_2 =   B_dC_d;
   \end{equation}

 \item Find $K_2\in \mathcal{L}(Y_1,X_2)$  to detect  system  $(A_2, C_{1\Lambda }\Gamma -C_2)$ exponentially;

 \item Set  $K_1= \Gamma K_2\in \mathcal{L}(Y_1,X_1)$.
\end{itemize}

 \begin{theorem}\label{Th20205201209}
Let $(A_1,B_1, C_1)$ be  a    regular linear system with the  state space $X_1$,
 input  space $U_1$ and the output space $Y_1$.
Suppose that the exosystem satisfies  \dref{2020681738},   system \dref{20204271052} is approximately observable, $B_d\in \mathcal{L}(U_d,[D(A_1^*)]')$  and
$A_1$ generates an exponentially stable $C_0$-semigroup $e^{A_1t}$ on $X_1$.
 If  the regulation problem \dref{20204271052} is solvable, i.e., the regulator equations \dref{20205201035}  admits a    solution   $\Pi \in \mathcal{L}(X_2,X_1)$
and $Q\in\mathcal{L}(X_2,U_1)$  \cite{Natarajan2016TAC},
then there exist $K_1\in \mathcal{L}(Y_1,X_1) $ and
$K_2\in \mathcal{L}(Y_1,X_2) $  such that
 the observer \dref{20205201013}  of system \dref{20204271052} is well-posed:
  For any $(\hat{z}_1(0),\hat{z}_2(0))^\top\in
 X_1\times X_2$ and $u\in L^2_{\rm loc}([0,\infty);U)$,
 the observer \dref{20205201013}    admits a unique solution $(\hat{z}_1,\hat{z}_2)^\top\in C([0,\infty);X_1\times X_2)$  such that
\begin{equation} \label{20205201210}
 e^{\omega t} \|(z_1(t)-\hat{z}_1(t), z_2(t)-\hat{z}_2(t))^{\top}\|_{X_1\times X_2}\to  0
  \ \ \mbox{as}\ \ t\to\infty,
\end{equation}
   where       $\omega$  is a positive constant that is independent of $t$.
\end{theorem}
\begin{proof}
  Since the semigroup $e^{A_1t}$ is exponentially stable in  $X_1$ and
$\sigma(A_2)\subset\{\lambda\ |\ {\rm Re}\  \lambda\geq 0\}$,
we have   $\sigma(A_1)\cap\sigma(A_2)=\emptyset$. By \cite{Natarajan2016TAC},
the Sylvester equation \dref{wxh202078226} admits a    unique  solution  $\Gamma\in \mathcal{L}(X_2,X_1)$.
In terms of the solution $\Pi \in \mathcal{L}(X_2,X_1)$, $Q\in\mathcal{L}(X_2,U_1)$ of the regulator
equations \dref{20205201035}, we can define the transformation \dref{20205201048} to
convert system \dref{20204271052}  into  system  \dref{20205201040}.
Moreover, a simple computation shows that
\begin{equation} \label{2020522926}
   \tilde{A}_1   Sx_2   - S   {A}_2x_2 =   B_1Qx_2 ,\ \  S =\Gamma+\Pi \in \mathcal{L}(X_2,X_1),\ \ \forall\ x_2\in X_2.
\end{equation}
 In view of the observer \dref{wxh202020617402},
 the observer of system \dref{20205201040} is
 \begin{equation} \label{20205201159}
 \left\{\begin{array}{l}
\disp \dot{\hat{x}}_1(t) =  A_1\hat{x}_1(t)+B_1Q\hat{x}_2(t)
-F_1[y(t)- C_{1 }\hat{x}_1(t)  ]+ B_1u(t),\crr
\disp  \dot{\hat{x}}_2(t) =A_2\hat{x}_2(t)+F_2[y(t)-C_{1 }\hat{x}_1(t)
 ],
\end{array}\right.
\end{equation}
 where  the tuning  gains   $F_1$ and $F_2$  satisfy:
  $F_2\in \mathcal{L}(Y_1,X_2)$  detects system $(A_2, C_{1 \Lambda }S)$ and
 $F_1= {S}F_2\in \mathcal{L}(Y_1,X_1)$.  By  exploiting Corollary \ref{Co2020425849},
for any initial state $(\hat{x}_1(0),\hat{x}_2(0)) ^{\top}\in  X_1\times X_2$,
 the observer \dref{20205201159}
 admits a unique solution $(\hat{x}_1,\hat{x}_2)^{\top}\in C([0,\infty);X_1\times X_2)$  such that
 \dref{wxh202020618082} holds for   some
 positive constant $\omega$.
  Let
   \begin{equation} \label{2020523721}
 \begin{pmatrix}
   \hat{z}_1(t) \\
   \hat{z}_2(t)
 \end{pmatrix}=
  \begin{pmatrix}
   I_1&  \Pi \\
   0&I_2
 \end{pmatrix}  \begin{pmatrix}
   \hat{x}_1(t) \\
   \hat{x}_2(t)
 \end{pmatrix},\ \ K_2=F_2\ \ \mbox{and}\ \ K_1= \Gamma K_2.
\end{equation}
Then,  it is easy to see  that  such a defined
 $(\hat{z}_1,\hat{z}_2)^{\top}\in C([0,\infty);X_1\times X_2)$
is a solution of the  observer \dref{20205201013} and moreover,  the convergence \dref{20205201210} holds.

Finally, we show that    both   $K_1$ and   $K_2$ are always bounded.
Indeed,  $K_2\in \mathcal{L}(Y_1,X_2)$  is trivial.
Since  $\Gamma\in \mathcal{L}(X_2,X_1)$ and $K_1=  \Gamma K_2$,   we have $K_1\in \mathcal{L}(Y_1,X_1)$.
This completes the proof of the theorem.
\end{proof}
By \dref{20205201013},  Theorem \ref{Th20205201209} and \cite{Natarajan2016TAC}, an observer
based  error   feedback control for  system \dref{20204271052} can be designed naturally as
\begin{equation} \label{2020851136}
u(t)=-Q\hat{z}_2(t),
\end{equation}
where $\hat{z}_2(t)$ comes from the observer \dref{20205201013}.
The exponential stability of the    resulting closed-loop system
can be obtained easily by the  separation principle of the linear systems.

\begin{remark}\label{Re20208301237}
Although  the
regulator equations
\dref{20205201035} have been used in the process of
    observer   design,
it is interesting  that the observer \dref{20205201013} itself is free from the regulator equations, in other words, we can  design an observer for  system \dref{20204271052} without solving the
  regulator equations \dref{20205201035}.
\end{remark}

\begin{remark}\label{Re2020841839}
Generally speaking, we need to consider the ``robustness" of the feedback \dref{2020851136}
in   robust output regulation.
This problem  can   be investigated  by the internal model principle \cite{Lassi2010SIAM}.
Since the robustness is beyond the topic of this paper, we do not   touch   it here.
However, the problem  is almost trivial provided $ Y_1 =\mathbb{C}$.
Indeed, a simple computation shows that   the  feedback  \dref{2020851136} contains 1-copy
internal model of the exosystem  and hence  is conditionally robust \cite{Lassi2017SIAM, Lassi2014SIAM}.
Consequently, the controller \dref{2020851136}  is robust to all unknown  $ {C}_2$ and $  C_d$ provided $ Y_1 =\mathbb{C}$.
In other words, both  $C_2$ and $C_d$   can be selected specially as in \cite{Meng}.
\end{remark}

\begin{remark}\label{Re2020692101}
When $B_1=B_d$ in  system \dref{20204271052}, the disturbance $C_dz_2(t)$ is in the control channel.
System \dref{20204271052} can be regarded as a stabilization problem with  the input
disturbance $C_dz_2(t)$ and the output
disturbance  $C_2z_2(t)$.
When the disturbance is estimated, it can be compensated directly by its estimation.
This is the main idea of the active disturbance rejection control (ADRC).
\end{remark}
\section{ODEs with output delay}\label{ODE+delay}
In this section, we suppose that $X_2=\mathbb{R}^n$,  $A_2 \in\mathcal{L}(X_2) $ and $U $ is a Hilbert space.
We will  validate the   theoretical results  through the  output delay compensation for    ODEs.  Consider the following  single output  system:
\begin{equation} \label{2020691600}
\left.\begin{array}{l}
\dot{x}_2 (t) = A_2 x_2 (t)+B_2u(t) ,\ \ \
y(t)=C_2 x_2 (t-\tau ),
\end{array}\right.
\end{equation}
where  $C_2 \in \mathcal{L}(X_2 ,\mathbb{R} )$ is the  observation operator,
$B_2 \in \mathcal{L}(U ,X_2 )$ is the control  operator, $u(t)$ is the control
and   $y(t)$ is  the  measurement with  delay $\tau>0$.
Let $ w(x,t)= C_2 x_2(t-x)$  for   $x\in[0,\tau ]$ and $t\geq \tau$.
Then,  system \dref{2020691600} can be written as
\begin{equation} \label{2020691606}
\left\{\begin{array}{l}
\disp  \dot{x}_2 (t)=A_2x_2(t) +B_2u(t), \crr
\disp     w_t(x,t)+ w_x(x,t)=0  ,\ \ x\in(0,\tau), \crr
\disp  w(0,t)=C_2 x_2 (t),  \crr
\disp y(t)= w(\tau ,t).
\end{array}\right.
\end{equation}
Define $ A_1  :D( A_1 )\subset L^2[0,\tau]\to L^2[0,\tau]$ by
\begin{equation} \label{2020691609}
 \left.\begin{array}{l}
 A_1  f=-f' ,\ \ \forall\ f\in D( A_1 )=\left\{f\in   H^1[0,\tau] \ |\ f (0)=0\right\},
\end{array}\right.
\end{equation}
the operator $B_1: \R\to [D(A_1^*)]'$  by $B_1q=q\delta(\cdot)$ for any $q\in\mathbb{R}$, where $\delta(\cdot)$ is the Dirac distribution,
and the operator $C_{1}:    D(A_{1}) \subset L^2[0,\tau] \to \mathbb{R} $ by
\begin{equation} \label{2020691614}
 C_{1}f  =f (\tau),\ \ \forall\ f \in D(A_1).
\end{equation}
Then, system \dref{2020691606} is put into   an abstract form:
 \begin{equation} \label{20120691616}
 \left\{\begin{array}{l}
 w_t(\cdot,t)= A_1  w (\cdot,t)+B_1C_2x_2(t),\crr
 \dot{x}_2 (t)=A_2x_2(t)+B_2u(t),\crr
 y(t)=C_1w(\cdot,t).
\end{array}\right.
\end{equation}
Since $A_1$ generates an exponentially stable $C_0$-semigroup $e^{A_1t}$ on $ L^2[0,\tau]$,
we need to solve the Sylvester equation $A_1S-SA_2=B_1C_2$ for the observer design.

Inspired by \cite{FPart1}, we    suppose that  $S: X_2\to L^2[0,\tau]$ takes the form:
  \begin{equation} \label{2020681057}
\disp    Sv  =\sum_{j=1}^{n}v_j {s}_j :=\langle  s ,v\rangle_{X_2}
,\ \ \forall\ v=(v_1,v_2,\cdots,v_n)^{\top}\in X_2,
\end{equation}
where $s: [0,\tau]\to\R^n$ is a vector-valued function given by
 $s(x)=(s_1(x),s_2(x),\cdots,s_n(x))^{\top} $ for any $x\in[0,\tau]$, $s_j\in L^2[0,\tau]$, $j=1,2,\cdots,n$.
Inserting \dref{2020681057} into the corresponding Sylvester equation    will lead to  a  vector-valued  ODE:
\begin{equation} \label{201206934}
\disp  s'(x )+A_2^*s(x )=0,  \ \
\disp s(0)=-C_2^*.
\end{equation}
 We solve   \dref{201206934}  to get  $s(x )=-e^{-A_2^* x} C^*_2 $, $x\in[0,\tau]$.
Hence, the  solution of the Sylvester equation $A_1S-SA_2=B_1C_2$  is found to be
\begin{equation} \label{2020691640}
 (Sv)(x)=-{C_2} e^{-A_2  x} v  ,\ \ \forall\ v\in X_2,\ x\in [0,\tau].
 \end{equation}
Moreover, $C_1Sx_2=-{C_2} e^{-A_2  \tau}x_2, \ \forall\ x_2\in X_2$. According to  the scheme of gain operators  choice  in   Section \ref{Observer1}, if   $F_2$   detects $(A_2,C_1S)$, we can choose  $F_1=SF_2=-{C_2} e^{-A_2  \cdot} F_2$.
In view of  \dref{wxh202020617402}, the  observer of system \dref{20120691616} is
\begin{equation} \label{2020623754}
\left\{\begin{array}{l}
\disp  \hat{w}_t(x,t)+\hat{w}_x(x,t)=C_2e^{-A_2x}F_2 [w(\tau,t)-\hat{w}(\tau,t)],\ \ x\in(0,\tau),\crr
\disp \dot{\hat{x}}_2(t) = A_2\hat{ x}_2(t) + F_2 [w(\tau,t)-\hat{w}(\tau,t)]+B_2u(t) ,\crr
\disp \hat{w}(0,t)=C_2\hat{x}_2(t).
\end{array}\right.
\end{equation}
In particular, if we choose $F \in \mathcal{L}( \mathbb{R},X_2)$ such that
 $A_2+FC_2 $ is Hurwitz, then the operator  $A_2+e^{A_2\tau}F C_2 e^{-A_2 \tau}=
  A_2-e^{A_2\tau}F C_1S $ is also Hurwitz due to the invertibility of
 $e^{ A_2 \tau}$. Then, we   can choose  $F_2=-e^{A_2\tau}F$
  and $F_1= C_2e^{A_2(\tau-x)}F$.  The   observer of system \dref{20120691616} turns to be
  \begin{equation} \label{2020691644}
\left\{\begin{array}{l}
 \disp  \hat{w}_t(x,t)+\hat{w}_x(x,t)=-C_2e^{A_2(\tau-x)}F [w(\tau,t)-\hat{w}(\tau,t)],\ \
 x\in(0,\tau),\crr
\disp \dot{\hat{x}}_2(t) = A_2\hat{ x}_2(t) -e^{A_2\tau}F [w(\tau,t)-\hat{w}(\tau,t)] +B_2u(t),\crr
\disp \hat{w}(0,t)=C_2\hat{x}_2(t),
\end{array}\right.
\end{equation}
which is the same as \cite{KrsticDelaybook}
and  \cite{SmyshlyaevKrstic2008SCL} where the PDE backstepping method was  exploited.
In contrast to the PDE backstepping, the target system is never needed here
and the  Lyapunov function is not  used in   stability analysis. This avoids the  difficulty of  target system choosing  and the Lyapunov-based technique
for PDEs with    time-delay.
More importantly, the  main idea of  developed approach can also be extended to the delay compensation for infinite-dimensional system, which will be considered in detail in the next paper  \cite{FPartDelay} of this series works.
\section{Unstable heat with ODE dynamics}\label{heat+ODE}
In this section, we consider a more complicated problem to show the effectiveness  of the proposed  approach.
We will observe an unstable heat equation through   ODE sensor dynamics.
Comparing with the ODE system with PDE sensor dynamics, the PDE system with ODE sensor dynamics is
more complicated.
The intuitive  reason behind  this  is  that we need to  observe  an  infinite-dimensional plant   through a  finite-dimensional system.

Consider the following  cascade system of ODEs and an unstable  heat equation:
 \begin{equation} \label{20206101017}
 \left\{\begin{array}{l}
 \disp \dot{v} (t) = A_1  v(t)+B_1w(1,t), \crr
\disp w_t(x,t)=w_{xx}(x,t)+\mu w(x,t) ,\     x\in(0,1),\  \mu>0,\crr
\disp  w(0,t)=0,\ \ w_x(1,t)= u(t),\crr
 y(t)= C_1 v(t),
\end{array}\right.
\end{equation}
 where $A_1 \in  \R^{m\times m} $, $B_1\in \mathcal{L}(\R,\R^{m})$,
 $C_1\in \mathcal{L}(\R^{m},\R)$,   $u(t)$ is the control and $y(t)$ is the  measured output.
The heat subsystem is the control plant and the ODE system serves as the sensor dynamics.
This makes the observation of  \dref{20206101017}    very  different from the existing results  in  \cite{Krstic2009SCL} and \cite{KrsticDelaybook}.
Define the operator  $ A_2  :D( A_2 )\subset L^2[0,1]\to L^2[0,1]$ by
\begin{equation} \label{20206101041}
 \left.\begin{array}{l}
 A_2  f =f'' +\mu f ,\ \ \forall\ f\in
D( A_2 )=\left\{f\in  H^2[0,1] \mid f(0)= f'(1)=0\right\},
\end{array}\right.
\end{equation}
the operator $B_2\in \mathcal{L}(\mathbb{R}, [D(A^*_2)]')$ by $B_2c:=\delta(\cdot-1)c$ for any $c\in\R$ and $C_2\in \mathcal{L}( D(A_2), \R  )$ by
\begin{equation} \label{20206101042}
 \left.\begin{array}{l}
 \disp C_2f = f(1),\ \ \forall\ f\in D(A_2),\end{array}\right.
\end{equation}
where $\delta(\cdot)$ is the Dirac distribution.
With these operators at hand,   system \dref{20206101017} can be written  as an abstract form in the state space $\R^m\times L^2[0,1]$:
\begin{equation} \label{20206101044}
\left\{\begin{array}{l}
  \disp \dot{v} (t) = A_1  v(t)+B_1C_2w(\cdot,t), \crr
 w_t(\cdot,t) =A_2w (\cdot,t)+B_2u(t),\crr
y(t)=C_1v(t).
\end{array}\right.
\end{equation}
According to Theorem \ref{wxhTh20202211035}, to design  the observer for system \dref{20206101044},
we should  first   choose $F_0 \in {\mathbb R}^{m}$ such that $A_1+F_0C_1$ is Hurwitz and then  solve
the following Sylvester equation  on $D(A_2)$:
\begin{equation} \label{20206101055}
  ( {A}_1+F_0C_{1  }) S   - S{A}_2 =   B_1C_{2 } .
   \end{equation}
Inspired by \cite{FPart1}, we define  the vector-valued function $s:[0,1]\to\R^m$ by
$s(x)=(s_1(x),s_2(x),\cdots,s_m(x))^{\top} $  for any $x\in[0,1]$, where $s_j\in    L^2[0,1]$ will be determined later, $j=1,2,\cdots,m$.
Suppose that  the solution of  Sylvester equation
  \dref{20206101055} takes the form
  \begin{equation} \label{20206101058}
\disp   Sf=
\begin{pmatrix}
\langle f,  s_1\rangle_{L^2[0,1]}\\
\langle f,  s_2\rangle_{L^2[0,1]}\\
\vdots\\
\langle f,  s_m\rangle_{L^2[0,1]}\\
\end{pmatrix}  :=\langle f,s\rangle_{L^2[0,1] }     ,\ \ \forall\ f\in
L^2[0,1].
\end{equation}
Then, inserting \dref{20206101058} into \dref{20206101055}, we obtain
\begin{equation}\label{wxh2020311804}
s''(x)+\mu s(x)=(A_1+F_0C_1)s(x),\ \ s(0)=0,\ s'(1)=B_1,
\end{equation}
which is  a vector-valued  ODE with respect to the variable $x$.
Solve equation \dref{wxh2020311804} to obtain
 \begin{equation}\label{wxh2020311806}
\disp s(x)=x\mathcal{G}(xG)(\cosh G )^{-1}B_1,\ \ G^2= A_1+F_0C_1  -\mu ,
\end{equation}
where
\begin{equation}\label{20208181020}
 \mathcal{G}(s)=\left\{\begin{array}{ll}
   \disp \frac{\sinh s}{s} ,& s\neq0,s\in \mathbb{C},\\
   1, & s=0.
 \end{array}\right.
   \end{equation}
By \dref{wxh2020311806} and  \dref{20206101058}, $S\in \mathcal{L}(L^2[0,1],\mathbb{R}^m)$ satisfies
\begin{equation} \label{wxh2020311818}
Sf=\int_0^1 f(x)s(x) dx ,\ \ \ \forall\ f \in L^2[0,1],
 \end{equation}
and hence $C_1S\in \mathcal{L}(L^2[0,1],\mathbb{R})$  given by
 \begin{equation} \label{wxh2020311819}
C_1Sf=
\int_{0}^{1}f(x)C_1s(x) dx=\langle  f,C_1s\rangle_{L^2[0,1]} ,\ \ \ \forall\  f\in L^2[0,1].
 \end{equation}
According to the developed scheme, we need to design  $F_2\in \mathcal{L}(\mathbb{R},[ D (A_2^*)]')$  to  detect  system $(A_2, C_{1}S)$ exponentially, i.e., design $F_2$  such that the following system is exponentially stable
\begin{equation} \label{wxh2020312229}
 \left\{\begin{array}{ll}
\disp   z _t(x,t)= z_{xx}(x,t)+\mu z(x,t)+F_2\int_{0}^{1}C_1s(x) z(x,t)dx,\crr
\disp   z(0,t)= z_x(1,t)=0.
\end{array}\right.
\end{equation}
We treat this problem by the modal decomposition approach  which has been used in    \cite{CoronTrelat2004SICON,PrieurandTrelat2019TAC,Russell1978SIAMReview}.

By a simple computation, system \dref{wxh2020312229} can be written as
\begin{equation}\label{wxh20203122292}
z_t(\cdot,t) = (A_2+F_2C_1S)z(\cdot,t).
\end{equation}
 Let
\begin{equation}\label{wxh2020321044}
\phi_n(x)=\sqrt{2}\sin \sqrt{\lambda_n}x , \ \ \lambda_n=\left(n-\frac{1}{2}\right)^2\pi^2, \ x\in[0,1],\ \ n= 1,2,\cdots.
\end{equation}
Then,  $\{\phi_n(\cdot)\}_{n=1}^{\infty} $ forms an orthonormal basis for  $L^2[0,1]$, which satisfies
\begin{equation}\label{wxh2020321047}\left.\begin{array}{l}
 \phi_n''(x)=-\lambda_n\phi_n(x),\ \
\phi_n(0)=\phi'_n(1)=0,\  n=1,2,\cdots.
\end{array}\right.\end{equation}
We represent   the solution   of system \dref{wxh2020312229} and the function
$C_1s$    as
\begin{equation}\label{wxh2019123042}\left.\begin{array}{l}
\disp    z(\cdot,t)=\sum\limits_{n=1}^{\infty}z_n(t)\phi_n(\cdot), \ \
C_1s(\cdot)=\sum\limits_{n=1}^{\infty}\gamma_n \phi_n(\cdot),
\end{array}\right.\end{equation}
where $z_n(t)$ and $\gamma_n$ are    the Fourier coefficients, given by
\begin{equation}\label{wxh201912304}\left.\begin{array}{l}
z_n(t)=\displaystyle \int_{0}^{1}z(x,t)\phi_n(x)dx ,\
\gamma_n =\displaystyle \int_{0}^{1}C_1s(x)\phi_n(x)dx,\ \ n=1,2,\cdots.
\end{array}\right.\end{equation}
Finding the derivative of $z_n(t)$ along the system  \dref{wxh2020312229},  we obtain
\begin{equation}\label{wxh201912305} \begin{array}{rl}
\dot{z}_n(t)=&\displaystyle \int_{0}^{1} z_t(x,t)\phi_n(x)dx\crr
=&\displaystyle \int_{0}^{1}\left[ z_{xx}(x,t)+\mu z (x,t)+F_2\int_{0}^{1}C_1s(x)z(x,t)dx \right]\phi_n(x)dx\crr
=&\displaystyle (-\lambda_n+\mu)z_n(t)
+ \left\langle F_2 \sum\limits_{j=1}^{\infty}\gamma_j z_j(t),\phi_n\right\rangle_{[D(A_2)]',D(A_2)}\crr
=&\displaystyle (-\lambda_n+\mu)z_n(t)
+ \left[\sum\limits_{j=1}^{\infty}\gamma_jz_j(t)\right]F_2^*\phi_n.
\end{array} \end{equation}
Since $\lambda_n\to +\infty (n\to \infty)$,
there exists a positive integer $N$ so that
\begin{equation}\label{wxh2020321759}
 (-\lambda_n+\mu)<0,\ \ \forall\ n>N.
 \end{equation}
 Suppose that there exists an  $L_N=[l_1,l_2,\cdots,l_N]^\top\in \mathbb{R}^{N} $ such that
 $\Lambda_N+L_N\Gamma_N$ is  Hurwitz, where
 \begin{equation}\label{wxh2020321800}
\Lambda_N={\rm diag}\{-\lambda_1+\mu,\cdots,-\lambda_N+\mu\},\  \;
\Gamma_N=[\gamma_1,\gamma_2,\cdots,\gamma_N] \in \mathbb{R}^{1\times N}.
\end{equation}
If we choose $ F_2\in L^2(\R,L^2[0,1]) $ by
 \begin{equation}\label{2020611806}
 F_2 c =c \sum\limits_{n=1}^{N}l_n\phi_n(\cdot) ,\ \ \forall\ c\in \R,
\end{equation}
then, it follows from Lemma \ref{wxh2020321655} in Appendix that  the  operator $ A_2+F_2C_1S  $ generates an exponentially stable $C_0$-semigroup  $e^{(A_2+F_2C_1S ) t}$ on $L^2[0,1]$. As a result, $F_2\in \mathcal{L}(\mathbb{R},L^2[0,1])$   detects  system $(A_2, C_{1}S)$ exponentially.

Let $F_1=F_0+SF_2$. It follows from  \dref{2020611806} and \dref{wxh2020311818} that
 \begin{equation}\label{wxh2020321018}
 F_1 c =\left(F_0+\int_0^1 s(x)  \sum\limits_{n=1}^{N}l_n \phi_n(x) dx\right) c,\  \  \forall\ c\in \mathbb{R}.
\end{equation}
In view of    \dref{wxh202020617402}, \dref{2020611806} and \dref{wxh2020321018}, the   observer of system \dref{20206101017} is designed as
 \begin{equation}\label{wxh20202272158}
\left\{\begin{array}{l}
\dot{\hat{v}} (t)= A_1 \hat{v}  (t)+B_1\hat{w}(1,t)
-\left[F_0+\disp\int_0^1 s(x)  \sum\limits_{n=1}^{N}l_n \phi_n(x) dx\right][C_1v(t)-C_1\hat{v} (t)],\crr
 \disp   \hat{w}_t(x,t)=\hat{w}_{xx}(x,t)+\mu \hat{w}(x,t)+\sum\limits_{n=1}^{N}l_n \phi_n(x) [C_1v(t)-C_1\hat{v} (t)]  ,\crr
\disp   \hat{w}(0,t)=0, \  \hat{w}_x(1,t)=u(t).
\end{array}\right.
\end{equation}

\begin{theorem}\label{wxhTh202032926}
Let $\phi_n(\cdot)$ and $\lambda_n$  be given  by \dref{wxh2020321044} and  let
   $s(\cdot)$ be given by \dref{wxh2020311806}.
   Suppose that system \dref{20206101017} is approximately observable
   and the inequality \dref{wxh2020321759} holds with the positive integer $N$.
   Then,  there exists an  $F_0 \in {\mathbb R}^{m}$   and   $L_N=[l_1,l_2,\cdots,l_N]^\top\in \mathbb{R}^{N} $ such that observer \dref{wxh20202272158}  of system \dref{20206101017} is well-posed:
For  any $(\hat{v}(0),\hat{w}(\cdot,0)) ^{\top}\in \R^m\times L^2[0,1]$ and
  $u\in L^2_{\rm loc}[0,\infty)$,
  the observer \dref{wxh20202272158} admits a unique solution    $(\hat{v} ,\hat{w} )^{ \top}\in C([0,\infty);\R^m\times L^2[0,1])$ such that
\begin{equation} \label{20206111426}
 e^{\omega t} \|(v(t)-\hat{v} (t), w(\cdot,t)-\hat{w} (\cdot,t)) ^\top \|_{\R^{ m}\times L^2[0,1]}\to  0
  \ \ \mbox{as}\ \ t\to\infty,
\end{equation}
where $\omega$  is a positive constant that is independent of $t$.
\end{theorem}
\begin{proof}
It is easy to verify that  $(A_2,B_2,C_2)$ is a  regular linear system, where $A_2,B_2,C_2$ are defined by \dref{20206101041}-\dref{20206101042} (see, e.g.\cite{ByrnesJDCS2002}).
By  \cite[Definition 1.2, p.3]{Higham2008book}, the matrix-valued function
$\mathcal{G}(xG)$ and $\cosh G$ are well-defined.
Since system \dref{20206101017} is approximately observable, it follows from  \cite[Remark 6.1.8, p.175]{TucsnakWeiss2009book} and the   cascaded structure of system  \dref{20206101017} that system $(A_1,C_1)$ is observable as well.
As a result, there exists an  $F_0\in \mathcal{L}(Y_1,X_1)$ such that $A_1+F_0C_1$ is Hurwtiz,
and at the same time, $\sigma({A}_1+F_0C_{1})\cap \sigma (A_2)=\emptyset$.
By a simple computation, it follows that
\begin{equation} \label{20208181115}
    \sigma(A_2-\mu)=\left\{ -\left(n-\frac{1}{2}\right)^2\pi^2\ \Big{|}\ n\in \mathbb{N}\right\},
\end{equation}
which, together with \dref{wxh2020311806}, leads to
   $\sigma(G^2)\cap  \left\{ -\left(n-\frac{1}{2}\right)^2\pi^2\ \Big{|}\ n\in \mathbb{N}\right\} =\emptyset$.
Hence, for any $\lambda \in \sigma(G)$, we have $\lambda^2 \notin \left\{ -\left(n-\frac{1}{2}\right)^2\pi^2\   \Big{|} \ n\in \mathbb{N}\right\}$,
which implies that  $\lambda \notin \left\{ i\left(n-\frac{1}{2}\right)\pi \   \Big{|} \ n\in \mathbb{Z}\right\}$ and hence $\cosh \lambda \neq 0$.
Consequently,   $\cosh G$ is invertible.

By a simple computation, the operator $S$, given by \dref{wxh2020311818} and \dref{wxh2020311806}, solves  the Sylvester equation \dref{20206101055} and moreover, $C_1S$  given by \dref{wxh2020311819} satisfies $C_1S\in \mathcal{L}(L^2[0,1],\mathbb{R})$.
Define $\A_{F_0} =\begin{pmatrix}
 {A}_1+F_0C_{1  }&B_1C_{2\Lambda }\\0& \tilde{A}_2
\end{pmatrix}$  and
$  \mathcal{C}_1 =( C_1,  0) \in \mathcal{L}( \R^{ m}\times L^2[0,1],\R) $. As
 the proof of Corollary \ref{Co2020425849},  it follows from
 \cite[Remark 6.5]{Weiss1994MCSS} that
  $(\A_{F_0}, \mathcal{C}_1)$  is approximately observable. By  Lemma   \ref{THwxh20202291137},  the pair  $(\mathbb{S}\A_{F_0} \mathbb{S}^{-1},\mathcal{C}_1\mathbb{S}^{-1})=\left( \begin{pmatrix}
{A}_1+F_0C_{1  }&0\\0& \tilde{A}_2
\end{pmatrix}, (C_1,-C_1S)\right)$  is approximately observable as well
where the invertible  transformation $\mathbb{S}$   is given by
   \begin{equation} \label{20208131626}
   \left.\begin{array}{l}
\disp    \mathbb{S}  \left(  v ,  f \right)^{\top}= \left( v +S f,\ f  \right)^{\top},\ \ \forall\ (  v ,f)^\top\in  \R^m\times L^2[0,1].
 \end{array}\right.
\end{equation}
Thanks to   the block-diagonal structure of $\mathbb{S}\A_{F_0} \mathbb{S}^{-1}$ and
\cite[Lemma 10.2]{FPart1},      system     $(A_2,  C_1S )$ is  approximately  observable.
 Since   $\{\phi_n(\cdot)\}_{n=1}^{\infty}$ defined  by \dref{wxh2020321044}  forms an orthonormal   basis for  $L^2[0,1]$,   we conclude  from \dref{wxh2020311819} and the  approximate observability of $(A_2, C_{1}S)$ that
\begin{equation}\label{202061118}\left.\begin{array}{l}
 \gamma_n =\displaystyle \int_{0}^{1}C_1s(x)\phi_n(x)dx\neq0,\ \ n=1,2,\cdots.
\end{array}\right.\end{equation}
As a result, the finite-dimensional system $(\Lambda_N,\Gamma_N)$  is observable,
  where  $\Lambda_N  $ and $  \Gamma_N $ are given  by \dref{wxh2020321800}.
Hence, there exists an  $L_N=[l_1,l_2,\cdots,l_N]^\top\in \mathbb{R}^{N} $ such that
 $\Lambda_N+L_N\Gamma_N$ is  Hurwitz.
By Lemma \ref{wxh2020321655}, $F_2\in \mathcal{L}(\mathbb{R},L^2[0,1])$
 given by \dref{2020611806}   detects  system $(A_2, C_{1}S)$ exponentially.
Noting that $F_1=F_0+SF_2$  satisfies \dref{wxh2020321018},   the    well-posedness of the
 observer \dref{wxh20202272158}   can be obtained by Theorem \ref{wxhTh20202211035} directly.
\end{proof}
\section{Numerical simulations} \label{simulation}
In this section, we carry out some numerical simulations for systems \dref{20206101017} and \dref{wxh20202272158} to validate the theoretical  results.
 The finite difference scheme is adopted in discretization.   The numerical results
are programmed in Matlab.
The time step   and the space step  are
taken as $4\times 10^{-5}$ and $ 10^{-2}$, respectively.
We choose
 \begin{equation}\label{wxh202082739}
A_1 = \begin{pmatrix}
 0 & -1\\
 1 & 0
\end{pmatrix},\ B_1=\begin{pmatrix}
  1\\
1
\end{pmatrix},\ C_1=(1,\ 1), \ \mu=4.
\end{equation}
With this setting, it is easy to check that the assumptions in Theorem \ref{wxhTh202032926}  are fulfilled with $N=1$.
 The initial states  of  systems \dref{20206101017} and \dref{wxh20202272158}   are chosen as
\begin{equation}\label{wxh2020824737}
  v(0)=(1,1)^\top, \  w(x,0)=\sin \pi x ,\  \hat{w}(x,0)\equiv0,
  \ \hat{v}(0)=0,\ \ x\in[0,1].
 \end{equation}
We assign  the poles  to get  the gains
\begin{equation}\label{wxh2020824739}
F_0=\begin{pmatrix}
  -1\\
-1
\end{pmatrix},\
 F_1=\begin{pmatrix}
   -1.5847\\
   -3.9479
\end{pmatrix},\
F_2=5.0978\times\sqrt{2}\sin \frac{\pi}{2}x,
\end{equation}
which  lead  to $\sigma(\Lambda_N+L_N\Gamma_N )=\{-2\}$.
The  error between the state  $w(\cdot,t)$ and $\hat{w}(\cdot,t)$  is   plotted
in Figure~\ref{Fig1}(a) and the  error  between  the  state   $v(t)=(v_1(t),v_2(t))^\top$ and $\hat{v}(t)=(\hat{v}_1(t),\hat{v}_2(t))^\top$   is  plotted in Figure~\ref{Fig1}(b).
Both of them   show that the convergence   is   effective  and smooth,
which validates numerically the effectiveness of the proposed method.


\begin{figure}[!htb]\centering
\subfigure[$w(x,t)-\hat{w}(x,t)$]
 {\includegraphics[width=0.48\textwidth]{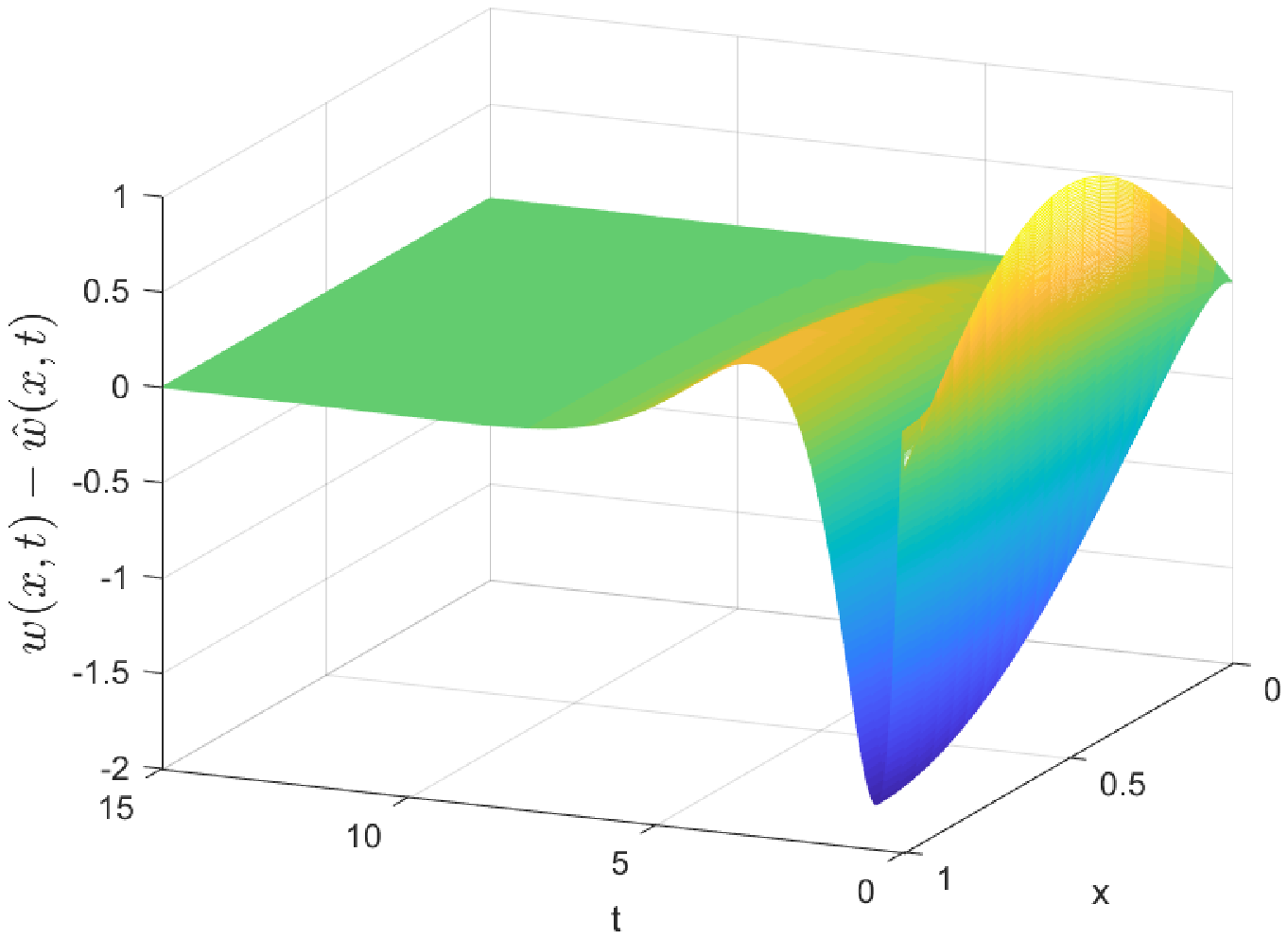}}
 \subfigure[ $v(t)-\hat{v}(t)$]
 {\includegraphics[width=0.48\textwidth]{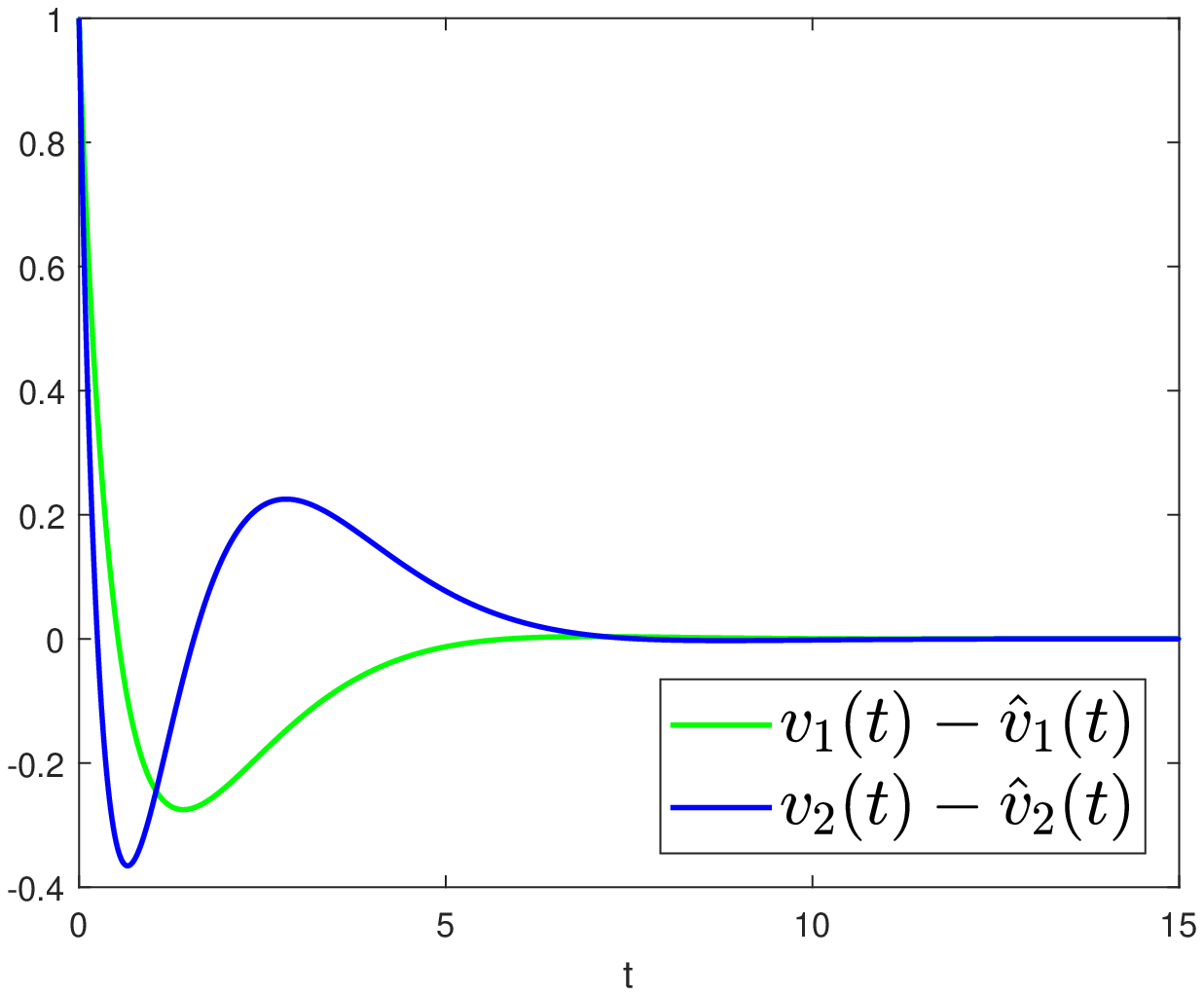}}
\caption{Observation error.}\label{Fig1}
\end{figure}
\section{Conclusions}\label{remarksensor}
This paper attributes the  observer design for abstract cascade   systems.
Two important   issues are addressed.
  The first one is the sensor  dynamics compensation for a linear system
  and the second one is on the error based observer   design in the problem of output regulation.
It is found that these two    different problems can be dealt with in a united way from the abstract linear cascade systems  framework  point of view.
  As a  result, the observer design approaches in sensor dynamics compensation  and the disturbance estimation approaches  in  output regulation can be   used interactively.
We propose a new scheme for   the observer design by seeking    tuning gain operators   of the  Luenberger-like observer.
  Both the  well-posedness and the exponential convergence of the observer    are established. The proposed method gives an alternative method for  observer design of the PDE  cascade systems, which avoids the target system seeking and the  Lyapunov function constructing
   in the popular PDE backstepping method.

It should be pointed out that the proposed approach in  Theorem  \ref{wxhTh20202211035} is not limited
to the examples considered in Sections \ref{ODE+delay} and \ref{heat+ODE}.
Actually,  the approach opens up a new road leading to the observer design of cascade systems particularly for
those systems which consist of ODE  and multi-dimensional PDE.
More importantly,
 it can also  be  applied to  the   delay  compensation for general infinite-dimensional systems. This will be considered  in
  the    third paper  \cite{FPartDelay}  of this series works.
 Moreover,
   as pointed in Remark \ref{Re2020692101},  the approach  gives rise to a new idea of  active disturbance rejection control.
In the last  paper \cite{FPart4} of this series works, we will give a new observer, i.e., {\it extended dynamics observer}, in the  framework of  the active disturbance rejection control.

\section{Appendix A}
\begin{lemma}\label{wxh2020321655}
Let the operator $A_2$ be given by \dref{20206101041},      let  $\phi_n(\cdot)$ and  $\lambda_n$ be given by \dref{wxh2020321044} and let   $N$ be  a positive integer such  that \dref{wxh2020321759} holds with $\mu>0$.
Suppose that there exists an  $L_N=[l_1,l_2,\cdots,l_N]^\top\in \mathbb{R}^{N} $ such that
 $\Lambda_N+L_N\Gamma_N$ is  Hurwitz, where $\Lambda_N  $ and $  \Gamma_N $ are given  by \dref{wxh2020321800}.
Suppose that   $F_2: \R\to L^2[0,1]$  is  defined by
\dref{2020611806}    and $C_1S$    is given by \dref{wxh2020311819}.
 Then,  the  operator $A_2+F_2C_1S$ generates an exponentially stable $C_0$-semigroup  $e^{(A_2+F_2C_1S) t}$ on $L^2[0,1]$.
\end{lemma}
\begin{proof}
By \dref{2020611806} and \dref{wxh2020311819}, the operator $F_2C_1S$ is bounded on $L^2[0,1]$.
Since $A_2 $ generates   an analytic semigroup on $L^2[0,1]$,  so is for the operator  $A_2+F_2C_1S$  due to the boundedness of $F_2C_1S$ (\cite[Corollary 2.3, p.81]{Pazy1983Book}).
The proof will be accomplished if we can show that   $\sigma(A_2+F_2C_1S)\subset \{ s\ | \ {\rm Re}s<0\}$.
For any $\lambda\in \sigma(A_2+F_2C_1S)$, we consider characteristic equation
$(A_2+F_2C_1S)f=\lambda f$ with $f\neq 0$.

When $f\in  {\rm Span}\{\phi_1,\phi_2,\cdots,\phi_N\}$,
assume that $f=\sum_{j=1}^Nf_j\phi_j$.
The characteristic equation becomes
\begin{equation}\label{wxh2020322242}
 \sum_{j=1}^Nf_jA_2\phi_j+\sum_{j=1}^N f_j F_2C_1S \phi_j= \sum_{j=1}^N\lambda f_j\phi_j.
\end{equation}
By \dref{2020611806} and \dref{wxh201912304}, we have
 \begin{equation}\label{20208131714}
F_2C_1S \phi_j= F_2  \int_{0}^{1}C_1s(x)\phi_j(x) dx
  =F_2 \gamma_j=\gamma_j \sum\limits_{n=1}^{N}l_n \phi_n,\ \ j=1,2,\cdots,N.
\end{equation}
 Since $A_2\phi_j=(-\lambda_j+\mu)\phi_j$,  equation \dref{wxh2020322242} takes the form
\begin{equation}\label{201912312141}
 \sum_{j=1}^Nf_j(-\lambda_j+\mu)\phi_j+ \sum_{j=1}^N f_j\gamma_j  \cdot\sum\limits_{n=1}^{N}l_n \phi_n= \sum_{j=1}^N\lambda f_j\phi_j.
\end{equation}
Take the  inner product with $\phi_k$, $k=1,2,\cdots,N$ on equation \dref{201912312141} to get
\begin{equation}\label{201912312146}
(-\lambda_k+\mu)f_k+l_k\sum_{j=1}^N f_j\gamma_j=  \lambda f_k,\ \ k=1,2,\cdots,N,
\end{equation}
which, together with \dref{wxh2020321800}, leads to
\begin{equation}\label{201912312148}
 (\lambda - \Lambda_N-L_N\Gamma_N) \begin{pmatrix}
                       f_1\\f_2\\\vdots\\f_N
                     \end{pmatrix} =0.
\end{equation}
Since $(f_1,f_2,\cdots,f_N)\neq 0$, we have
\begin{equation}\label{201912312149}
{\rm Det}(\lambda - \Lambda_N-L_N\Gamma_N) =0,
\end{equation}
which shows that  $\lambda\in \sigma( \Lambda_N+L_N\Gamma_N) \subset\{ s\ | \ {\rm Re}s<0\}$, since    $  \Lambda_N+L_N\Gamma_N$ is Hurwitz.

When $f\notin {\rm Span}\{\phi_1,\phi_2,\cdots,\phi_N\}$,
there must exist $j_0>N$ so that $\displaystyle \int_{0}^{1}f(x)\phi_{j_0}(x)dx\neq 0$.
Take the  inner product with $\phi_{j_0}$ on equation $(A_2+F_2C_1S)f=\lambda f$ to get
 \begin{equation}\label{wxh2020322257}
(-\lambda_{j_0}+\mu) \int_{0}^{1}f(x)\phi_{j_0}(x)dx=  \lambda \int_{0}^{1}f(x)\phi_{j_0}(x)dx.
\end{equation}
As a result,   $ \lambda =-\lambda_{j_0}+\mu<0$.
This shows also  $\lambda\in \sigma (A_2+F_2C_1S)\subset \{ s\ | \ {\rm Re}s<0\}$.
The proof is complete.
\end{proof}
\end{document}